\documentclass[ejs]{imsart}

\RequirePackage{amsthm,amsmath,amsfonts,amssymb}
\RequirePackage[numbers]{natbib}
\RequirePackage[colorlinks,citecolor=blue,urlcolor=blue]{hyperref}
\RequirePackage{graphicx, multirow, algorithm, algcompatible}

\arxiv{2310.00864}
\startlocaldefs
\theoremstyle{plain}

\newtheorem{theorem}{Theorem}[section]
\newtheorem{lemma}[theorem]{Lemma}
\newtheorem{proposition}{Proposition}
\theoremstyle{definition}

\theoremstyle{remark}

\DeclareMathOperator*{\argmax}{arg\,max}
\DeclareMathOperator*{\argmin}{arg\,min}
\newcommand{\indep}{\perp \!\!\! \perp}
\endlocaldefs

\begin{document}
\begin{frontmatter}
\title{Multi-Label Residual Weighted Learning for Individualized Combination Treatment Rule}
\runtitle{ICTR}

\begin{aug}
\author[A]{\fnms{Qi}~\snm{Xu}\ead[label=e1]{qxu6@uci.edu}\orcid{0000-0003-1127-6523}},
\author[A]{\fnms{Xiaoke}~\snm{Cao}\ead[label=e2]{xiaokec@uci.edu}},
\author[B]{\fnms{Geping}~\snm{Chen}\ead[label=e3]{gepingc@iastate.edu}\orcid{0009-0004-4551-1653}},
\author[C]{\fnms{Hanqi}~\snm{Zeng}\ead[label=e4]{hanqizeng@hsph.harvard.edu}\orcid{0009-0001-4156-7260}},
\author[D]{\fnms{Haoda}~\snm{Fu}\ead[label=e5]{fu\_haoda@lilly.com}}
\and
\author[A]{\fnms{Annie}~\snm{Qu}\ead[label=e6]{aqu2@uci.edu}\orcid{0000-0002-8396-7828}}
\address[A]{Department of Statistics,
University of California, Irvine\printead[presep={,\ }]{e1,e2,e6}}

\address[B]{Department of Statistics,
Iowa State University\printead[presep={,\ }]{e3}}

\address[C]{Department of Biostatistics, Harvard University\printead[presep={,\ }]{e4}}

\address[D]{Eli Lilly and Company\printead[presep={,\ }]{e5}}

\runauthor{Xu. Q et al.}

\end{aug}

\begin{abstract}
Individualized treatment rules (ITRs) have been widely applied in many fields such as precision medicine and personalized marketing. Beyond the extensive studies on ITR for binary or multiple treatments, there is considerable interest in applying combination treatments. This paper introduces a novel ITR estimation method for combination treatments incorporating interaction effects among treatments. Specifically, we propose the generalized $\psi$-loss as a non-convex surrogate in the residual weighted learning framework, offering desirable statistical and computational properties. Statistically, the minimizer of the proposed surrogate loss is Fisher-consistent with the optimal decision rules, incorporating interaction effects at any intensity level - a significant improvement over existing methods. Computationally, the proposed method applies the difference-of-convex algorithm for efficient computation. Through simulation studies and real-world data applications, we demonstrate the superior performance of the proposed method in recommending combination treatments.
\end{abstract}



\end{frontmatter}
\section{Introduction}

The individualized treatment rule (ITR) in precision medicine has been widely applicable in recommending tailored treatment for each individual. Unlike the traditional one-size-fits-all strategy, ITR aims to account for subject heterogeneity to achieve personalization. While existing ITR methods mainly focus on choosing one of two or more treatments, combination treatments have emerged as a promising strategy to achieve better outcomes including enhanced efficacy and resistance prevention \citep{pernas2018balixafortide, mokhtari2017combination, kalra2010combination, maruthur2016diabetes}. Furthermore, combination treatments have also been applied in personalized marketing, where a mix of promotional strategies are tailored to diverse customer groups. In summary, developing ITR methods for combination treatments is of great interest not only in precision medicine and personalized marketing, but also potentially in many other fields.

The existing literature on ITR estimation can be broadly summarized into two categories. The first of these categories is the indirect approach, including the well-known Q-learning \citep{qian2011performance, clifton2020q}, D-learning \citep{qi2018d, qi2020multi} and A-learning \citep{shi2018high, lu2013variable, schulte2014q}. These methods propose parametric or non-parametric models for conditional average treatment effects to recommend preferable treatment. Another mainstream of ITR is the direct approach, including outcome-weighted learning \citep{zhao2012estimating, zhou2018outcome, zhang2020multicategory, liu2021outcome, xue2022multicategory}, and residual-weighted learning \citep{zhou2017residual}. These methods directly maximize the value function with respect to the decision rules. In practice, direct methods have demonstrated superior empirical performance, as they circumvent model misspecification issues common to indirect approaches. In addition, their decision rules are flexible, accommodating either parametric models (e.g., linear decision rule \citep{zhao2012estimating, zhang2020multicategory}) or nonparametric models (e.g., kernel method \citep{zhao2012estimating, zhang2020multicategory}, neural network \citep{liang2018estimating}, and boosting \citep{wang2020boosting}).

In regard to the combination treatment problem, we can apply the aforementioned multicategory ITRs to recommend the combination treatment problems, where each combination is treated as an independent treatment. Therefore, the correlations among combination treatments are ignored by the multicategory ITRs. Consequently, as the number of single treatments increases, the number of combination treatments tends to increase exponentially, which leads the model complexity of either direct or indirect approaches to explode. Given the limited or moderate sample sizes in biomedical applications, the estimation efficiency of multicategory ITRs is severely compromised. To address this issue, recent work \citep{xu2023optimal} proposed an indirect approach which estimates the conditional average treatment effects (CATE) with the double encoder model. This method has been shown to achieve both empirically and theoretically efficient estimation for combination treatments. Among the direct approaches, \citep{liang2018estimating, ye2023stage} utilized the Hamming loss in place of the 0-1 loss, treating the ITR estimation as a weighted multi-label classification problem. Under this formulation, we only need to estimate the decision rules for each single treatment, which greatly reduces the model complexity and addresses the inefficiency issue in multicategory ITRs. However, a parsimonious model may be incapable of incorporating interaction effects among combination treatments. Specifically, \citep{liang2018estimating} could be undermined when the interaction effects are non-negligible. Therefore, it is essential to develop direct methods that offer flexible modeling and are capable of incorporating interaction effects among combination treatments.

In this paper, we introduce a novel Multi-Label Residual Weighted Learning (MLRWL) framework for estimating the optimal ITR for combination treatments. Specifically, we propose using the generalized $\psi$-loss as a non-convex surrogate for the 0-1 loss in the multi-label classification problem, with the optimal ITR derived as the minimizer of the weighted generalized $\psi$-loss. The proposed method has two main advantages over the Hamming hinge loss considered in \citep{liang2018estimating}. First, the generalized $\psi$-loss guarantees the Fisher consistency, regardless of whether interaction effects are present. In particular, the minimizer of the weighted generalized $\psi$-loss exhibits sign consistency with the optimal ITR, a property that holds at any intensity level of interaction effects. This property is especially valuable in real applications, where the intensity of interaction effects could be unknown. Second, the generalized $\psi$-loss can accommodate negative or shifted outcomes to stabilize the empirical performance. In theory, we demonstrate that the Fisher consistency and the consistency of the proposed estimator are preserved given negative and shifted outcome weights. In contrast, a convex surrogate loss, such as the Hamming hinge loss, can only accommodate positive weights to preserve its convexity, potentially limiting its applicability.

Computationally, the non-convex generalized $\psi$-loss can be formulated as the difference between two convex functions. Therefore, the minimization of the weighted generalized $\psi$-loss can be solved efficiently by the difference of the convex (DC) algorithm \citep{tao1988duality} iteratively. Notably, the subproblem within each iteration is a quadratic programming problem for both linear and nonlinear decision rules, which can be solved by the quadratic programming solver. We show that estimators obtained through the DC algorithm are  stationary points. Our numerical studies indicate that the proposed method achieves superior performance compared with existing ITR approaches for combination treatment problems.

The rest of the article is organized as follows. In Section 2, we introduce the background of the ITR problem and existing works. In Section 3, we propose Multi-Label Residual Weighted Learning with the generalized $\psi$-loss. Algorithms and implementation details for linear and nonlinear decision rules are also illustrated. In Section 4, the theoretical properties of the generalized $\psi$-loss are provided. In Sections 5 and 6, we present numerical studies to evaluate the empirical performance of the proposed method in simulation settings and a real application to a type-2 diabetes study.

\color{black}
\section{Background}
In this paper, we focus on estimating an individualized treatment rule (ITR) for combination treatments using clinical experiment data. The variables of interest are $(\mathbf{X}, \mathbf{A}, Y)$, where $\mathbf{X} \in \mathcal{X} \subset \mathbb{R}^{p}$ represents pre-treatment covariates, $\mathbf{A}= (A^{(1)}, A^{(2)}, ..., A^{(K)}) \in \mathcal{A} = \{-1, 1\}^{K}$ denotes the combination treatments consisting of up to $K$ single treatments, and $Y\in \mathbb{R}$ is the observed outcome. We assume that a larger $Y$ indicates a more desirable outcome, and use $Y(\mathbf{A})$ to represent the potential outcome \citep{rubin1974estimating} under the treatment assignment $\mathbf{A}$. Since only one of the potential outcomes can be observed for each subject, it is infeasible to  recommend the subject-wise optimal treatment. Instead, our goal is to learn an ITR $d(\cdot): \mathcal{X} \rightarrow \mathcal{A}$ through maximizing the average outcome across the population. Here, the expected potential outcome under ITR $d(\cdot)$ is also termed as the value function \citep{qian2011performance} with respect to $d(\cdot)$:
\begin{align}
\label{value_func}
    \mathcal{V}(d) = \mathbb{E}[Y\{d(\mathbf{X})\}].
\end{align}

To estimate the ITR $d(\cdot)$ from clinical experiments, we rely on the following standard causal assumptions \citep{hernan2010causal}: \\
(a) Stable Unit Treatment Value Assumption (SUTVA): $Y = Y(\mathbf{A})$; \\
(b) No unmeasured confounders: $\mathbf{A} \indep Y(\mathbf{a})|\mathbf{X}$ for any $\mathbf{a} \in \mathcal{A}$; \\
(c) Positivity: $\mathbb{P}(\mathbf{A} = \mathbf{a}|\mathbf{X})\ge p_{\mathcal{A}} > 0$ for any $\mathbf{a} \in \mathcal{A}$, $\mathbf{X} \in \mathcal{X}$. \\
Under these causal assumptions, the optimal ITR $d^*(\cdot)$ satisfies
\begin{align}
    \label{value_func2}
    d^*(\cdot) = \argmax_{d(\cdot)}\mathcal{V}(d) = \argmax_{d(\cdot)}\mathbb{E}\left[\frac{Y}{\mathbb{P}(\mathbf{A}|\mathbf{X})}\mathbb{I}(\mathbf{A} = d(\mathbf{X}))\right],
\end{align}
where $\mathbb{I}(\cdot)$ is the indicator function and $\mathbb{P}(\mathbf{A}|\mathbf{X})$ is the propen
sity score \citep{hirano2003efficient}. Moreover, maximizing the value function (\ref{value_func2}) is equivalent to minimizing the following risk:
\begin{align}
    \label{loss_func}
    d^*(\cdot) = \argmin_{d(\cdot)}\mathcal{R}(d)
    = \argmin_{d(\cdot)}\mathbb{E}\left[\frac{Y}{\mathbb{P}(\mathbf{A}|\mathbf{X})}\mathbb{I}(\mathbf{A} \neq d(\mathbf{X}))\right], 
\end{align}
which is equivalent to a weighted classification problem, with $\mathbf{A}$ being the response comprised of $2^K$ distinct classes and weights given by $Y/\mathbb{P}(\mathbf{A}|\mathbf{X})$. However, directly minimizing the risk (\ref{loss_func}) is an NP-hard problem due to the non-smoothness of the indicator function $\mathbb{I}(\cdot)$. More intricate than binary or multicategory treatment problems, minimizing (\ref{loss_func}) for combination treatments encounters the curse of the dimensionality issue. As the number of single treatments $K$ grows, the number of possible combination treatments grows exponentially, which requires a rather complex model $d(\cdot)$ as the decision rule. Consequently, the estimation efficiency is undermined in combination treatment problems, especially those with a large $K$.

Since the combination treatment $\mathbf{A}$ can be considered as a $K$-dimensional binary response, it is natural to approach the problem (\ref{loss_func}) from the multi-label classification perspective \citep{liang2018estimating, tsoumakas2007multi}. Each treatment $A^{(k)}$ can be treated as a binary response, indicating whether the $k$th treatment was assigned ($A^{(k)}=1$) or not ($A^{(k)}=-1$). Therefore, we can decompose the ITR $d(\cdot)$ into $K$ decision rules: $d^{(1)}(\cdot), ..., d^{(K)}(\cdot)$, where $d^{(k)}(\cdot)$ decides whether the $k$th treatment should be assigned or not. In contrast to the multicategory classification requiring $2^{K}$ decision rules, $K$ decision rules are sufficient under the multi-label classification framework.

There are two mainstream strategies to tackle multi-label classification in the literature: the first strategy is the so-called binary relevance \citep{luaces2012binary, boutell2004learning}, which treats each label $A^{(k)}$ as an independent binary label and builds independent binary classifiers for each label. In combination treatment problems, the combination of multiple treatments could potentially induce additional interaction effects. These effects can be either synergistic or antagonistic effects in nature. As a result, the outcome $Y$ is largely contingent on the holistic treatment assignment $\mathbf{A}$, rather than solely on the individual assignment $A^{(k)}$. Therefore, it is risky to adopt the binary relevance strategy in combination treatment problems which may ignore the considerable interaction effects.
 
Second, it is prevalent to propose an appropriate loss function, particularly convex surrogate losses, to replace the 0-1 loss. The convexity property promotes an efficient computation algorithm that guarantees global optimality. However, the improved computational efficiency may come at the cost of compromised statistical properties. For instance, \citep{liang2018estimating} combines the Hamming and hinge losses as a convex surrogate loss to tackle the combination treatment problem. However, the minimizer of their surrogate loss does not guarantee a Fisher-consistent estimation of the optimal ITR $d^*(\cdot)$, especially when significant interaction effects exist among treatments. In summary, it is critical to identify a suitable surrogate loss that guarantees sound statistical properties such as Fisher consistency while achieving efficient computation.

\section{Methodology}

In Section \ref{sec: mlrwl}, we introduce the proposed residual weighted learning framework from the perspective of multi-label classification. Section \ref{sec: alg_imp} introduces the algorithm and implementation details for linear and nonlinear decision rules, respectively.

\subsection{Multi-Label Residual Weighted Learning}
\label{sec: mlrwl}

In this section, we introduce a novel non-convex surrogate loss, the generalized $\psi$-loss, that targets the weighted multi-label classification to estimate the optimal ITR for combination treatments. 

Specifically, the generalized $\psi$-loss associated risk, named as $\psi$-risk, is defined as follows:
\begin{align}
\label{psi_loss_func}
\begin{split}
    \mathcal{R}_{\psi}(f) &= \mathbb{E}\left[\frac{Y}{\mathbb{P}(\mathbf{A}|\mathbf{X})}\psi(Z^{(1)}, ..., Z^{(K)})\right] \\
    &= \mathbb{E}\left[\frac{Y}{\mathbb{P}(\mathbf{A}|\mathbf{X})}\{T_1(Z^{(1)}, ..., Z^{(K)}) - T_0(Z^{(1)}, ..., Z^{(K)})\}\right],
\end{split}
\end{align}
where $f(\cdot) = (f^{(1)}(\cdot), ..., f^{(K)}(\cdot))$ and $f^{(k)}(\cdot): \mathcal{X}\rightarrow\mathbb{R}$ to represent the decision function for the $k$th treatment, and the $k$th decision rule is followed by $d^{(k)}(\mathbf{X}) = \text{sign}(f^{(k)}(\mathbf{X}))$. The intermediate variables $Z^{(k)}$'s are defined as $Z^{(k)} = A^{(k)}f^{(k)}(\mathbf{X})$, where $Z^{(k)} > 0$ indicates that the $k$th label $A^{(k)}$ is correctly classified, and $Z^{(k)} \le 0$ signifies misclassification. 

The proposed generalized $\psi$-loss is denoted as $\psi(\ldots)$, composed of two parts with the same form $T_s(Z^{(1)}, ..., Z^{(K)}) = \max(s - Z^{(1)}, ..., s - Z^{(K)}, 0)$ for $s = 0, 1$. This loss is a generalization of the $\psi$-loss \citep{liu2006multicategory} which targets the binary or multicategory classification. The first term, $T_1(Z^{(1)}, ..., Z^{(K)})$, is an extension of the hinge loss to multi-label settings, which can be also formulated as follows:
\begin{align}
    T_1(Z^{(1)}, \ldots, Z^{(K)}) &= \max(1 - Z^{(1)}, \ldots, 1-Z^{(K)}, 0) \notag \\
    &= \max_{k}\bigg\{\max(1 - Z^{(k)}, 0)\bigg\}, \notag
\end{align}
which characterizes the largest hinge loss over all labels. We also term $T_1(\ldots)$ as the generalized hinge loss for (weighted) multi-label classification, which is a convex surrogate loss of the 0-1 loss. The second term has the same shape as $T_1(\ldots)$ but passes through the origin. Subtracting $T_0(\ldots)$ from $T_1(\ldots)$ is equivalent to truncating the generalized hinge loss $T_1(\ldots)$ at 1 if $Z^{(k)} \ge 0$ for any $k \in \{1, ..., K\}$. A visual comparison of the generalized hinge loss and the generalized $\psi$-loss in $2$-label scenarios is shown in Figure \ref{fig: loss_comp}. 

\begin{figure}[!t]
    \centering
    \includegraphics[width=300pt]{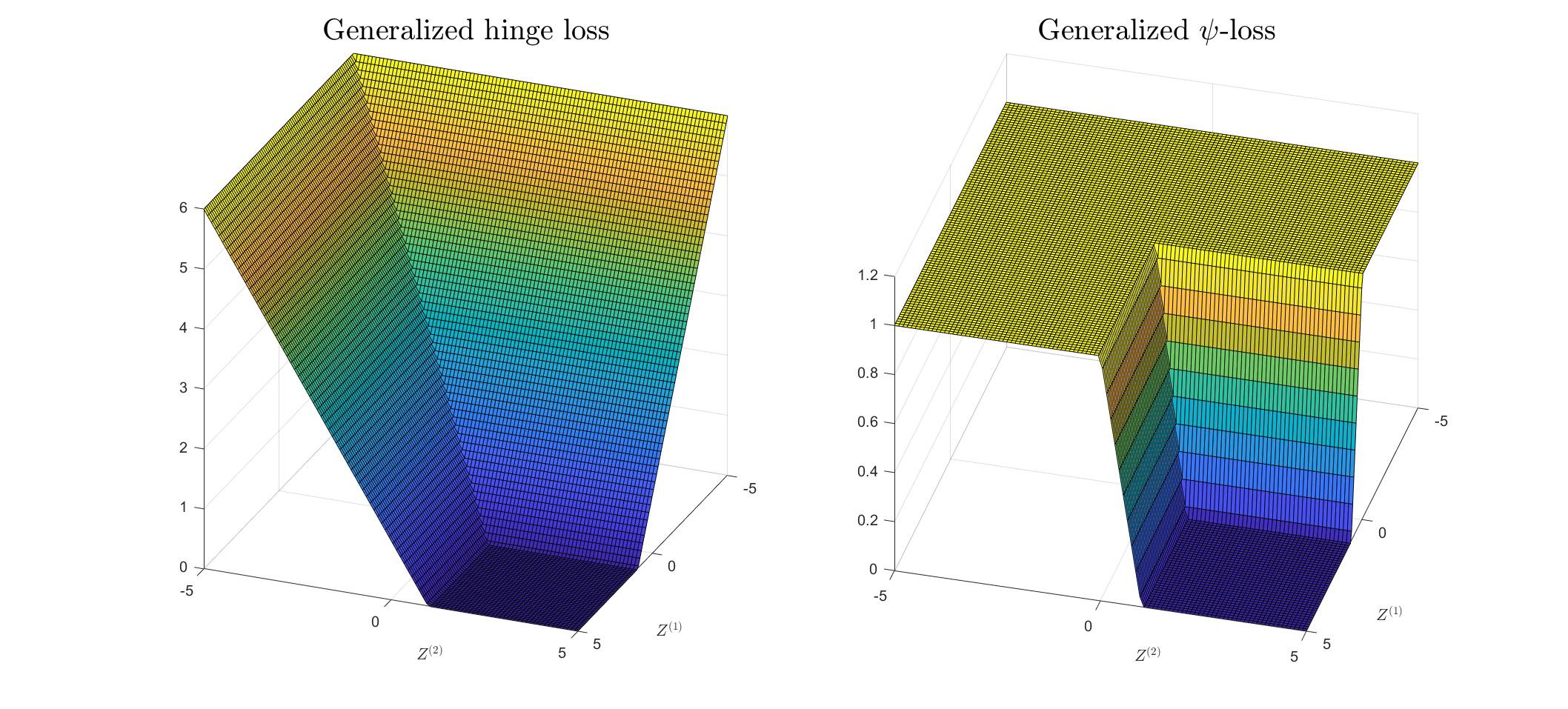}
    \caption{Illustration and comparison of the generalized hinge loss and the generalized $\psi$-loss in the 2-label classification scenario.}
    \label{fig: loss_comp}
\end{figure}

As shown in Figure \ref{fig: loss_comp}, the generalized $\psi-$loss $\psi(Z^{(1)}, ..., Z^{(K)}) = 0$ if and only if $Z^{(k)} > 1$ for all $k = 1, ..., K$; conversely, $\psi(Z^{(1)}, ..., Z^{(K)}) = 1$ if any one of the $Z^{(k)}$'s is negative. In other words, the generalized $\psi-$loss is minimized only if all the decision functions $f^{(k)}(\cdot)$'s perfectly assign the observed treatments:
\begin{align}
   |f^{(k)}(\mathbf{X})| > 1 \text{ and sign}(f^{(k)}(\mathbf{X})) = A^{(k)} \text{ for all }k = 1, ..., K. \notag
\end{align}
In order to minimize the $\psi-$risk (\ref{psi_loss_func}), $\psi(Z^{(1)}, ..., Z^{(K)})$ is expected to be minimal for large weights $\frac{Y}{\mathbb{P}(\mathbf{A}|\mathbf{X})}$, and the decision functions $f^{(k)}(\cdot)$'s are expected to align with the associated treatment assignment $\mathbf{A}$. Therefore, the overall treatment effects of the combination treatments, including treatment effects of single treatments and the induced interaction effects, affect the decision rules simultaneously. This property indeed guarantees the Fisher consistency of the proposed method, irrespective of the intensity of interaction effects. In contrast, each single treatment contributes to the Hamming hinge loss \citep{liang2018estimating}, leading the decision rules to rely more on the treatment effects of single treatments, so the Fisher consistency of their method is only achieved with minor interaction effects.

\subsubsection{Outcome Shift}
\label{sec: outcome_shift}
A significant drawback of the outcome-weighted learning framework \citep{zhao2012estimating, liang2018estimating} is that the value of $Y$ must be positive to preserve the convexity of the surrogate loss function. Empirically, it is possible to shift the outcome so that the assumption is satisfied; however, the shift of the outcome may impact the empirical performance of the algorithm. We refer readers to \citep{zhou2017residual} for a detailed discussion about this potential issue. 

The proposed generalized $\psi$-loss is also sensitive to the shift of the outcome $Y$, and the impact of the outcome shift for the combination treatment problems is rather significant. Specifically, the minimizer of the $\psi$-risk under outcome weights and the minimizer of the $\psi$-risk under shifted outcome weights $Y - g(\mathbf{X})$'s are not necessarily equivalent, because: $\mathcal{R}_{\psi}(f) \neq \mathcal{R}_{\psi, g}(f) + \text{constant}$ where $\mathcal{R}_{\psi, g}(f)$ is defined as follows:
\begin{align}
    \mathcal{R}_{\psi, g}(f) = \mathbb{E}\left[\frac{Y-g(\mathbf{X})}{\mathbb{P}(\mathbf{A}|\mathbf{X})}\psi(Z^{(1)}, ..., Z^{(K)})\right],
\end{align}
where $g(\mathbf{X})$ is a measurable function. Therefore, minimizing the $\psi$-risk given the shifted outcome weights $Y - g(\mathbf{X})$ might result in decision rules other than the optimal decision rules derived from (\ref{psi_loss_func}). Therefore, selecting an appropriate shift is crucial for maintaining statistical consistency and optimality of the decision rules. 

In this work, we employ the treatment-free effects $g(\mathbf{X})=\frac{1}{|\mathcal{A}|}\sum\mathbb{E}[Y|\mathbf{X}, \mathbf{A}]$ as a functional shift, and the inverse probability weighted residual $\frac{Y-g(\mathbf{X})}{\mathbb{P}(\mathbf{A}|\mathbf{X})}$ is regarded as the weight in the multi-label classification. There are two main reasons for choosing $g(\mathbf{X})$ as the treatment-free effects: First, it does not change the relative orders of the conditional average treatment effects over all possible treatments:
\begin{align}
    &\mathbb{E}[Y|\mathbf{X}, \mathbf{A}=\mathbf{a}] - \mathbb{E}[Y|\mathbf{X}, \mathbf{A}=\mathbf{a}'] \notag \\
    = &\mathbb{E}[Y-g(\mathbf{X})|\mathbf{X}, \mathbf{A}=\mathbf{a}] - \mathbb{E}[Y-g(\mathbf{X})|\mathbf{X}, \mathbf{A}=\mathbf{a}'], \quad \text{for all }\mathbf{a}, \mathbf{a}' \in \mathcal{A}, \notag
\end{align}
which is a sufficient condition to guarantee the Fisher consistency property, which will be elaborated in Section \ref{sec: theory}. Second, it leads to a straightforward interpretation: for treatments associated with above-average treatment effects, the decision rules are expected to match these treatments; for the treatment associated with below-average treatment effects, the decision rules are instead encouraged to deviate from these treatments.

Empirically, given the i.i.d samples $(\mathbf{x}_i, \mathbf{a}_i, y_i)_{i=1}^{n}$, we can estimate the ITR by minimizing 
\begin{align}
\label{empirical_loss}
    \min_{f^{(1)}, ..., f^{(K)}}\sum_{i=1}^{n}\frac{y_i - g(\mathbf{x}_i)}{\mathbb{P}(\mathbf{a}_i|\mathbf{x}_i)}\psi(z_i^{(1)}, ..., z_i^{(K)}) + \mathcal{P}_{\lambda}(f),
\end{align}
where $g(\cdot)$ is the treatment-free effects, $\mathbb{P}(\cdot|\cdot)$ is the propensity score in the clinical trial experiment. In observational study, both treatment-free effects and true propensity scores are unknown to us, so working models of $g(\mathbf{x})$ and $\mathbb{P}(\mathbf{a}|\mathbf{x})$ are needed. In Appendix \ref{A.8}, we discuss the estimation of working models which can be plugged into (\ref{empirical_loss}) to estimate decision rules. The penalty function $\mathcal{P}_{\lambda}(f)$ determines the function space of $f$ and controls its complexity. More detailed discuss of model specification is introduced in Section \ref{sec: alg_imp}. Once we obtain the estimated $\hat{f}^{(1)}, ..., \hat{f}^{(K)}$, the estimated ITR $\hat{d}$ is given by:
\begin{align}
    \hat{d}^{(k)}(\mathbf{x}) = \text{sign}(\hat{f}^{(k)}(\mathbf{x})). \notag
\end{align}

\subsection{Algorithm and Implementation}
\label{sec: alg_imp}

In this section, we introduce the algorithm and implementation details of the proposed method, under both the linear and nonlinear decision rules, respectively. Since the generalized $\psi$-loss is a non-convex surrogate loss, commonly adopted convex algorithms are not applicable. Nevertheless, the generalized $\psi$-loss enables a decomposition which can be represented as the difference of two convex functions, in which the difference of the convex algorithm \citep{nam2017minimizing, tao1988duality} is applicable for efficient computation. 

Suppose the decision functions $f^{(1)}, ..., f^{(K)}$ are parameterized by $\boldsymbol\beta_{k}$'s respectively. Then the empirical loss (\ref{empirical_loss}) can be reformulated as 
\begin{align}
\label{loss_decomp}
\begin{split}
    \mathcal{L}(\boldsymbol\beta) &= \sum_{i=1}^{n}w_i\bigg\{T_1(\boldsymbol\beta; \mathbf{x}_i, \mathbf{a}_i) - T_0(\boldsymbol\beta; \mathbf{x}_i, \mathbf{a}_i)\bigg\} + \frac{\lambda}{2}\sum_{k=1}^{K}\mathcal{P}(\boldsymbol\beta_k)  \\
    &= \underbrace{\frac{\lambda}{2} \sum_{k=1}^{K}\mathcal{P}(\boldsymbol\beta_k) + \sum_{i=1}^{n}|w_i|\bigg\{T_1(\boldsymbol\beta; \mathbf{x}_i, \mathbf{a}_i)\mathbb{I}(w_i \ge 0) + T_0(\boldsymbol\beta; \mathbf{x}_i, \mathbf{a}_i)\mathbb{I}(w_i < 0)\bigg\}}_{\text{Convex part}: \mathcal{L}_{\text{cvx}}}  \\
    &+ \underbrace{\sum_{i=1}^{n}-|w_i|\bigg\{T_1(\boldsymbol\beta; \mathbf{x}_i, \mathbf{a}_i)\mathbb{I}(w_i < 0) + T_0(\boldsymbol\beta; \mathbf{x}_i, \mathbf{a}_i)\mathbb{I}(w_i \ge 0)\bigg\}}_{\text{Concave part}: \mathcal{L}_{\text{cave}}}, 
\end{split}
\end{align}
where $\boldsymbol\beta = (\boldsymbol\beta_1, ..., \boldsymbol\beta_K)$ is the collection of parameters of the decision functions, and the weight $\frac{y_i - \hat{m}(\mathbf{x}_i)}{\hat{\mathbb{P}}(\mathbf{a}_i|\mathbf{x}_i)}$ is denoted as $w_i$ for ease of notation. The penalty function $\mathcal{P}(\boldsymbol\beta_k)$ is a convex penalty function associated with the type of the decision rule, and $\lambda$ serves as the tuning parameter for the penalty. 

Since the loss function $\mathcal{L}(\boldsymbol\beta)$ can be decomposed into two parts, $\mathcal{L}_{\text{cvx}}(\boldsymbol\beta)$ and $\mathcal{L}_{\text{cave}}(\boldsymbol\beta)$, we can employ the difference of the convex algorithm \citep{nam2017minimizing, tao1988duality} to obtain the estimation of $\boldsymbol\beta$. Specifically, at the $t$ th iteration, the subproblem is minimizing a linear minorization \citep{liu2005multicategory} of the loss function $\mathcal{L}(\boldsymbol\beta)$:
\begin{align}
    \boldsymbol\beta^{(t)} = \argmin_{\boldsymbol\beta} \mathcal{L}_{\text{cvx}} (\boldsymbol\beta) + <\nabla_{\boldsymbol\beta}\mathcal{L}_{\text{cave}}(\boldsymbol\beta^{(t-1)}), \boldsymbol\beta - \boldsymbol\beta^{(t-1)}>, \notag
\end{align}
where $\nabla_{\boldsymbol\beta}\mathcal{L}_{\text{cave}}(\boldsymbol\beta^{(t-1)})$ is the sub-gradient of $\mathcal{L}_{\text{cave}}(\boldsymbol\beta)$ at the iterated $\boldsymbol\beta^{(t-1)}$, and $<\cdot, \cdot>$ denotes the inner product. The algorithm for minimizing $\mathcal{L}(\boldsymbol\beta)$ is summarized as follows:

\begin{algorithm}
\caption{Difference of convex algorithm for minimizing $\mathcal{L}(\boldsymbol\beta)$}
\label{dc_alg}
\begin{algorithmic}
\State Initialize $\boldsymbol\beta^{(0)}$, set maximum iteration $T$
\FOR{$t = 1, 2, ..., T$}
    \STATE Compute the subgradients $\nabla_{\boldsymbol\beta}\mathcal{L}_{\text{cave}}(\boldsymbol\beta^{(t-1)})$
    \STATE Update $\boldsymbol\beta^{(t)}$ by solving the convex optimization problem:
    \STATE \hspace{1cm} $\min \mathcal{L}_{\text{cvx}} (\boldsymbol\beta) + <\nabla_{\boldsymbol\beta}\mathcal{L}_{\text{cave}}(\boldsymbol\beta^{(t-1)}), \boldsymbol\beta - \boldsymbol\beta^{(t-1)}>$
\ENDFOR
\STATE Output $\boldsymbol\beta^{(t)}$ if $\sum_{k=1}^{K}\lVert \beta^{(t)}_{k} - \beta^{(t-1)}_{k}\rVert_2 \le \epsilon$; where $\epsilon$ is a pre-specified threshold; Otherwise, output $\boldsymbol\beta^{(T)}$
\end{algorithmic}
\end{algorithm}

The above algorithm guarantees that the convergent point obtained from the iterations is the stationary point of $\mathcal{L}(\boldsymbol\beta)$ if the initial value $\boldsymbol\beta$ satisfies certain conditions, the detailed result and conditions can be found in the Appendix \ref{A.3}.

In order to minimize (\ref{loss_decomp}), which includes complex $\max$ operator and indicator functions,  we introduce the slack variables $\eta_i$'s to convert the $T_1(\boldsymbol\beta; \mathbf{x}_i, \mathbf{a}_i)$ and $T_0(\boldsymbol\beta; \mathbf{x}_i, \mathbf{a}_i)$ into linear constraints, then the convex optimization problem within each iteration of Algorithm \ref{dc_alg} is equivalent to the following problem:

\begin{align}
\label{primal_quad_prob}
\scriptsize
\begin{split}
    \min_{\boldsymbol\beta} & \quad \sum_{k=1}^{K}\mathcal{P}(\boldsymbol\beta_k) + \gamma \sum_{i=1}^{n}|w_i|\eta_i + \gamma \sum_{k=1}^{K}<\nabla_{\boldsymbol\beta_{k}}\mathcal{L}_{\text{cave}}(\hat{\boldsymbol\beta}^{(t-1)}), \boldsymbol\beta_k> \\
    s.t. & \quad \eta_i \ge \mathbb{I}(w_i > 0) - a_i^{(k)}f^{(k)}(\mathbf{x}_i), \quad \text{for any } k = 1,...,K \\
    & \quad \eta_i \ge 0;
\end{split}
\end{align}
where $\gamma$ is a constant depending on $\lambda$, and $\nabla_{\beta_{k}}\mathcal{L}_{\text{cave}}(\boldsymbol\beta)$ is the subgradient of the concave part $\mathcal{L}_{\text{cave}}$ with respect to $\boldsymbol\beta_{k}$:
\begin{align}
{\scriptsize
\label{subgrad}
\nabla_{\boldsymbol\beta_{k}}\mathcal{L}_{\text{cave}}(\boldsymbol\beta) = 
\begin{cases}
\sum_{i=1}^{n}\mathbb{I}(w_i \ge 0)|w_i|\nabla_{\boldsymbol\beta_k}f^{(k)}(\mathbf{x}_i), \\
\hspace{1.5cm} \text{if } k = \argmax_{l}\{- a_i^{(l)}f^{(l)}(\mathbf{x}_i)\} \text{ and }1 - a_i^{(k)}f^{(k)}(\mathbf{x}_i) > 0,  \\
\sum_{i=1}^{n}\mathbb{I}(w_i < 0)|w_i|\nabla_{\boldsymbol\beta_k}f^{(k)}(\mathbf{x}_i), \\
\hspace{1.5cm} \text{if } k = \argmax_{l} \{- a_i^{(l)}f^{(l)}(\mathbf{x}_i)\}\text{ and } - a_i^{(k)}f^{(k)}(\mathbf{x}_i) > 0, \\
0, \hspace{1.5cm} \text{otherwise}.
\end{cases}
}
\end{align}

In the following, we provide the implementation details for linear decision rules in Section \ref{sec: ldr}, and then generalize it to the nonlinear decision rules in Section \ref{sec: nldr}.

\subsubsection{Linear Decision Rule for Optimal ITR}
\label{sec: ldr}

Consider the linear decision rules $f^{(k)}(\mathbf{x})$ as follows:
\begin{align}
    f^{(k)}(\mathbf{x}) = \beta_{0k} + \mathbf{x}^T\boldsymbol\beta_{1k}, \notag
\end{align}
where $\boldsymbol\beta_{1k}\in\mathbb{R}^{p}$ and $\beta_{0k}\in\mathbb{R}$. Then the associated ITR $d(\cdot)$ assigns a subject with $\mathbf{x}$ to the $k$th treatment if $\beta_{0k} + \mathbf{x}^T\boldsymbol\beta_{1k} > 0$ and does not assign the $k$th treatment otherwise. For the linear decision functions, we define the penalty function as the Euclidean norm $\mathcal{P}(\boldsymbol\beta_k) = \lVert\boldsymbol\beta_{1k}\rVert_2^2$. Then, the convex programming problem (\ref{primal_quad_prob}) can be rewritten as 
\begin{align}
\label{primal_linear_prob}
\scriptsize
\begin{split}
    \min_{\boldsymbol\beta} & \quad \frac{1}{2}\sum_{k=1}^{K}\lVert\boldsymbol\beta_{1k}\rVert_2^2 + \gamma \sum_{i=1}^{n}|w_i|\eta_i + \gamma \sum_{k=1}^{K}<\nabla_{\boldsymbol\beta_{1k}}\mathcal{L}_{\text{cave}}(\hat{\boldsymbol\beta}^{(t-1)}), \boldsymbol\beta_{1k}> +  \\
    &\gamma \sum_{k=1}^{K}\nabla_{\beta_{0k}}\mathcal{L}_{\text{cave}}(\hat{\boldsymbol\beta}^{(t-1)})\beta_{0k} \\
    s.t. & \quad \eta_i \ge \mathbb{I}(w_i > 0) - a_i^{(k)}(\beta_{0k} + \mathbf{x}_i^T\boldsymbol\beta_{1k}), \quad \text{for any } k = 1,...,K \\
    & \quad \eta_i \ge 0,
\end{split}
\end{align}
which is a quadratic programming with decision function parameters $\boldsymbol\beta_{1k}$'s, $\beta_{0k}$'s and slack variables $\eta_i$'s, where $\eta_i$'s are associated with individual-wise linear constraints. In the scenarios with large sample size $n$ and relatively small dimension of covariates $p$, it is computationally efficient to solve the primal form (\ref{primal_linear_prob}) directly. Otherwise, it is preferable to solve the dual form by introducing the Lagrange multipliers $\theta_{ik}$'s. Specifically, the dual form of (\ref{primal_linear_prob}) can be formulated as follows:
\begin{align}
\label{dual_linear_prob}
\scriptsize
\begin{split}
    \min_{\lambda} &\quad \frac{1}{2}\sum_{k=1}^{K}\sum_{i=1}^{n}\sum_{j=1}^{n}\theta_{ik}\theta_{jk}a_{i}^{(k)}a_{j}^{(k)}\mathbf{x}_i^T\mathbf{x}_j - \gamma \sum_{k=1}^{K}\sum_{i=1}^{n}\theta_{ik}a_{i}^{(k)}\mathbf{x}_i^T\nabla_{\boldsymbol\beta_{1k}} \mathcal{L}_{\text{cave}}(\hat{\boldsymbol\beta}^{(t-1)}) -\\
    &\sum_{k=1}^{K}\sum_{i=1}^{n}\theta_{ik}I(w_i\ge 0) \\
    s.t. &\quad \sum_{k=1}^{K}\theta_{ik} \le \gamma |w_i|, \quad \gamma \nabla_{\beta_{0k}}\mathcal{L}_{\text{cave}}(\hat{\boldsymbol\beta}^{(t-1)}) = \sum_{i=1}^{n}\theta_{ik}a_{ik}, \quad \theta_{ik}\ge 0, 
\end{split}    
\end{align}
which can be solved by quadratic programming solvers such as Gurobi \citep{gurobi}. Given the estimated $\theta_{ik}$'s, we can derive the estimated $\boldsymbol\beta_{1k}$'s and $\beta_{0k}$'s, where the detailed derivations are provided in the Appendix \ref{A.1}. 

\subsubsection{Nonlinear Decision Rule for Optimal ITR}
\label{sec: nldr}

In the following, we consider the nonlinear decision functions as follows:
\begin{align}
    f^{(k)}(\mathbf{x}) = \beta_{0k} + \sum_{i=1}^{n}\mathcal{K}(\mathbf{x}, \mathbf{x}_i)\beta_{ik}, \notag
\end{align}
where $\mathcal{K}(\cdot, \cdot)$ is a valid kernel function associated with a reproducing kernel Hilbert space $\mathcal{H}_{\mathcal{K}}$. Therefore, $f^{(k)}(\mathbf{x})$ can represent nonlinear functions embedded by $\mathcal{H}_{\mathcal{K}}$ with a shift $\beta_{0k}$. The norm in $\mathcal{H}_{\mathcal{K}}$, denoted as $\lVert\cdot\rVert_{\mathcal{K}}$, is induced by the inner product:
\begin{align}
    <f, g>_{\mathcal{K}} = \sum_{i=1}^{n}\sum_{j=1}^{m}\alpha_i\beta_j\mathcal{K}(\mathbf{x}_i, \mathbf{x}_j), \notag
\end{align}
for $f(\cdot) = \sum_{i=1}^{n}\alpha_i\mathcal{K}(\mathbf{x}_i, \cdot)$ and $g(\cdot) = \sum_{j=1}^{m}\beta_j\mathcal{K}(\mathbf{x}_j, \cdot)$. When we plug $f^{(k)}(\mathbf{x})$ into (\ref{primal_quad_prob}), the convex programming in the $t$-th iteration is formulated as follows:
\begin{align}
\label{primal_nonlinear_prob}
\scriptsize
\begin{split}
\min_{\boldsymbol\beta}&\quad\frac{1}{2}\sum_{k=1}^{K}\sum_{i=1}^{n}\sum_{j=1}^{n}\beta_{ik}\beta_{jk}\mathcal{K}(\mathbf{x}_i, \mathbf{x}_j) + \gamma\sum_{i=1}^{n}|w_i|\eta_i + \gamma \sum_{k=1}^{K}\sum_{i=0}^{n}<\nabla_{\beta_{ik}}\mathcal{L}_{\text{cave}}(\hat{\boldsymbol\beta}^{(t-1)}), \beta_{ik}> \\
s.t. &\quad \eta_i \ge \mathbb{I}(w_i > 0) - a_i^{(k)}(\beta_{0k} + \sum_{j=1}^{n}\mathcal{K}(\mathbf{x}_j, \mathbf{x}_i)\beta_{jk}), \quad \text{for any } k = 1, ..., K\\
&\quad \eta_i \ge 0.
\end{split}
\end{align}
Even though the $f^{(k)}(\mathbf{x})$ belongs to an infinite-dimensional space, it is computationally efficient to consider the dual form of the problem, which is also called the kernel trick \citep{hastie2009elements}. After the Lagrange multipliers $\theta_{ik}$'s are introduced, the dual form is formulated as follows:
\begin{align}
\label{dual_nonlinear_prob}
\scriptsize
\begin{split}
    \min_{\lambda} &\quad \frac{1}{2}\sum_{k=1}^{K}\sum_{i=1}^{n}\sum_{j=1}^{n}\theta_{ik}\theta_{jk}a_{i}^{(k)}a_{j}^{(k)}\mathcal{K}(\mathbf{x}_i, \mathbf{x}_j) - \gamma \sum_{k=1}^{K}\sum_{i=1}^{n}\theta_{ik}a_{i}^{(k)}\nabla_{\beta_k} \mathcal{L}_{\text{cave}}(\hat{\boldsymbol\beta}^{(t-1)}) -\\
    &\sum_{k=1}^{K}\sum_{i=1}^{n}\theta_{ik}I(w_i\ge 0) \\
    s.t. &\quad \sum_{k=1}^{K}\theta_{ik} \le \gamma |w_i|, \quad \gamma \nabla_{\beta_{0k}}\mathcal{L}_{\text{cave}}(\hat{\boldsymbol\beta}^{(t-1)}) = \sum_{i=1}^{n}\theta_{ik}a_{i}^{(k)}, \quad \theta_{ik}\ge 0,
\end{split}
\end{align}
which can also be solved by quadratic programming solvers. 

\section{Theoretical Properties}
\label{sec: theory}

In this section, we develop the theoretical properties of the estimated ITR under the weighted multi-label classification framework. Particularly, we establish the Fisher consistency under the proposed generalized $\psi$-risk, to guarantee that the optimizer of the $\psi$-risk is theoretically optimal. Furthermore, we also establish the excess risk bound and the consistency of the estimator within the reproducing kernel Hilbert space.

First, we establish the Fisher consistency under the outcome weighted learning framework. Specifically, the following result holds:

\begin{lemma}
\label{lemma: fisher_consistency}
For any measurable function $f: \mathcal{X}\rightarrow \mathbb{R}^{K}$, if $\hat{f}$ minimizes the $\psi$-risk $\mathcal{R}_{\psi}(f)$, then $d^*(\mathbf{x}) = \text{sign}(\hat{f}(\mathbf{x}))$, where $d^{*}(\cdot)$ is the optimal ITR given in (\ref{loss_func}).
\end{lemma}

Lemma \ref{lemma: fisher_consistency} provides the validity of using the generalized $\psi$-loss as the surrogate loss in the outcome weighted learning framework to estimate the optimal ITR. More importantly, there is no requirement for the intensity of the interaction effects among combination treatments as in \citep{liang2018estimating}, which is an advantage in estimating the ITR for combination treatments. As we emphasized in Section \ref{sec: mlrwl}, $\psi$-loss can incorporate interaction effects of any intensity, and guarantees the above property.

In the following, we show that the Fisher consistency holds under the residual weighted learning framework:

\begin{theorem}
\label{thm: fisher_consistency}
For $g(\mathbf{X}) = \frac{1}{|\mathcal{A}|}\sum_{\mathbf{A}\in\mathcal{A}}\mathbb{E}[Y|\mathbf{X}, \mathbf{A}]$, the minimizer $\hat{f}$ of the surrogate risk $\mathcal{R}_{\psi, g}(f)$ satisfies that $\text{sign}(\hat{f}(\mathbf{x})) = d^{*}(\mathbf{x})$ where $d^{*}(\mathbf{x})$ is the optimal ITR as in Lemma \ref{lemma: fisher_consistency}.
\end{theorem}

Theorem \ref{thm: fisher_consistency} guarantees that the generalized $\psi$-loss is a valid surrogate loss in the sense of Fisher consistency. As we mentioned in Section \ref{sec: outcome_shift}, an arbitrary choice of $g(\mathbf{X})$ may violate the Fisher consistency. Our choice of treatment-free effects retains the relative order of the CATE among all treatments so that Fisher consistency is also guaranteed. The detailed proofs of Lemma 1 and Theorem 1 are provided in Appendix \ref{A.4} and \ref{A.5}. 

Next, we establish the relationship between the excess risk under the proposed generalized $\psi$-loss and the 0-1 loss.

\begin{theorem}
\label{thm: excess_risk} 
For $f = (f^{(1)}, ..., f^{(K)})$ and any measurable $f^{(k)}: \mathcal{X}\rightarrow \mathbb{R}$, and any probability distribution for $(\mathbf{X}, \mathbf{A}, Y)$, we have 
\begin{align}
    \mathcal{R}\{\text{sign}(f)\} - \mathcal{R}^* \le \mathcal{R}_{\psi, g}(f) - \mathcal{R}^{*}_{\psi, g},
\end{align}
where the $\mathcal{R}_{\psi, g}$ denotes the risk $\mathcal{R}_{\psi, g}$ with $g(\cdot)$ as treatment-free effects. 
\end{theorem}

Theorem \ref{thm: excess_risk} shows that the excess risk of any measurable decision functions $f$ under the 0-1 loss is no larger than the excess risk under the $\psi$-risk. This suggests that if we estimate the ITR by minimizing $\mathcal{R}_{\psi, g}(f)$, the risk of the minimizer $\hat{f}$ is close to the Bayes risk. 

However, the above theoretical analyses are all based on the population-level probability of $(\mathbf{X}, \mathbf{A}, Y)$. In practice, we are concerned more about the estimator obtained from the empirical distribution. In the following, we establish the consistency of the proposed estimator $\hat{f}_n$ learned from the empirical distribution with sample size $n$.

\begin{theorem}
\label{thm: consistency}
Suppose the penalty coefficient $\lambda$ in the primal form (\ref{loss_decomp}) satisfies $\lambda\rightarrow 0$ and $n\lambda\rightarrow \infty$. The weights $\frac{|Y - g(\mathbf{X})|}{\mathbb{P}(\mathbf{A}|\mathbf{X})}$'s are assumed to be upper bounded by some positive constant $M$ almost surely. Then for any distribution $P$ for $(\mathbf{X}, \mathbf{A}, Y)$, we have
\begin{align}
    \mathbb{P}\bigg\{\lim_{n\rightarrow\infty}\mathcal{R}_{\psi, g}(\hat{f}_n) = \inf_{f\in\mathcal{H}_{\mathcal{K}} + \{1\}}\mathcal{R}_{\psi, g}(f)\bigg\} = 1,
\end{align}
where $\hat{f}_n$ is the minimizer of the empirical loss (\ref{empirical_loss}) with sample size $n$, and $\mathcal{H}_{\mathcal{K}} + \{1\}$ denotes the shifted reproducing kernel Hilbert space we considered in Section \ref{sec: nldr}.
\end{theorem}

Theorem \ref{thm: consistency} claims that the risk of the proposed estimator obtained by minimizing the empirical risk (\ref{empirical_loss}) can converge in probability to the minimal of the population risk as sample size increases. In other words, the proposed estimator is consistent corresponding to the optimal decision rules for the combination treatments. In addition, Theorem \ref{thm: consistency} also holds for linear decision rules with a pre-specified linear kernel $\mathcal{K}(\cdot, \cdot)$. The technical details of the proof are provided in Appendix \ref{A.7}.
Furthermore, we provide the extension of Theorem \ref{thm: consistency} to observational study in Appendix \ref{A.9}.

\section{Numerical Studies}
\label{sec: sim}
In this section, we assess the performance of the proposed method through simulation studies which mimic real-world scenarios. In these simulations, we consider the treatment effects with varying complexities, while interaction effects of different intensities are included in all settings.

The simulation settings are conducted under different sample sizes (n = 400, 800, 2000). The pre-treatment covariates $\mathbf{X}\in\mathbb{R}^{10}$ are sampled uniformly from the interval $(-1, 1)$. In all designed simulation studies, combination treatments $\mathbf{A}$ are uniformly randomly assigned. We consider three data generation processes. In the first two settings, two treatments are considered ($K=2$), resulting in four possible combination treatments with different treatment effects. In the third setting, three treatments are considered ($K=3$), corresponding to a total of eight possible combinations, with more complex treatment effects involving non-linear combinations of the covariates. The detailed treatment effects specifications are described as follows:

\begin{itemize}
    \item[] Simulation setting 1:
    \begin{itemize}
        \item[] $\tau_{(-1, -1)}(\mathbf{X}) = 0;$ $\tau_{(-1, 1)}(\mathbf{X}) = 6 \cdot \mathbb{I}(X_1 + X_2 > 0) \cdot \mathbb{I}(-X_1 + X_2 < 0);$
        \item[] $\tau_{(1, -1)}(\mathbf{X}) = 5 \cdot \mathbb{I}(X_1 + X_2 < 0) \cdot \mathbb{I}(-X_1 + X_2 < 0);$
        \item[] $\tau_{(1, 1)}(\mathbf{X}) = 3 \cdot \mathbb{I}(X_1 + X_2 > 0) \cdot \mathbb{I}(-X_1 + X_2 > 0).$
    \end{itemize}
    \item[] Simulation setting 2:
    \begin{itemize}
        \item[] $\tau_{(-1, -1)}(\mathbf{X}) = (X_1 + X_2)^2;$ $\tau_{(-1, 1)}(\mathbf{X}) = X_2^2 + X_3X_4;$
        \item[] $\tau_{(1, -1)}(\mathbf{X}) = -X_3X_4;$ $\tau_{(1, 1)}(\mathbf{X}) = X_2^2 + 3X_5X_6.$
    \end{itemize}
    \item[] Simulation setting 3:
    \begin{itemize}
        \item[] $\tau_{(-1, -1, -1)}(\mathbf{X}) = 0;$ $\tau_{(-1, -1, 1)}(\mathbf{X}) = 2(X_1 + \exp(X_2));$ 
        \item[] $\tau_{(-1, 1, -1)}(\mathbf{X}) = X_3 + (X_4 + X_5)^2;$ 
        \item[] $\tau_{(-1, 1, 1)}(\mathbf{X}) = 2(X_1 + \exp(X_2)) + X_3 + (X_4 + X_5)^2 + \log((X_5 + 1)^2);$
        \item[] $\tau_{(1, -1, -1)}(\mathbf{X}) = \exp(X_6 + X_7);$ 
        \item[] $\tau_{(1, -1, 1)}(\mathbf{X}) = \exp(X_6 + X_7) + 2(X_1 + \exp(X_2)) + X_8 + X_9 + X_{10};$
        \item[] $\tau_{(1, 1, -1)}(\mathbf{X}) = \exp(X_6 + X_7) + X_3 + (X_4 + X_5)^2;$
        \item[] $\tau_{(1, 1, 1)}(\mathbf{X}) = \exp(X_6 + X_7) + X_3 + (X_4 + X_5)^2 + 2(X_1 + \exp(X_2)) + (X_1 - X_5 + X_6)^2.$
    \end{itemize}
\end{itemize}

After generating the treatment effects, we design the outcome of interest $Y$ as follows: 
\begin{align}
    Y = g(\mathbf{X}) + \tau_{\mathbf{A}}(\mathbf{X}) + \epsilon, \quad
    g(\mathbf{X}) = 1 + X_1 + 2X_2, \quad \epsilon\sim N(0, 0.3), \notag
\end{align}
where $g(\mathbf{X})$ is the treatment-free effects, and $\epsilon$ is the random noise. It is noteworthy that interaction effects among combination treatments are designed in all of the above simulation settings. Specifically, in simulation setting 1, we split a two-dimensional plane into four quadrants. Except for the quadrant $\mathbb{I}(X_1 + X_2 < 0, -X_1+X_2 > 0)$, the combination of two treatments induces either positive or negative interaction effects in the other three quadrants as shown in Figure \ref{fig: sim1_interaction_effects}. In simulation settings 2 and 3, the interaction effects are polynomials and nonlinear functions of the pre-treatment covariates $X$, respectively. Additionally, the true decision rules are linear in simulation setting 1, and nonlinear in simulation settings 2 and 3. 

\begin{figure}
    \centering
    \includegraphics[width=300pt]{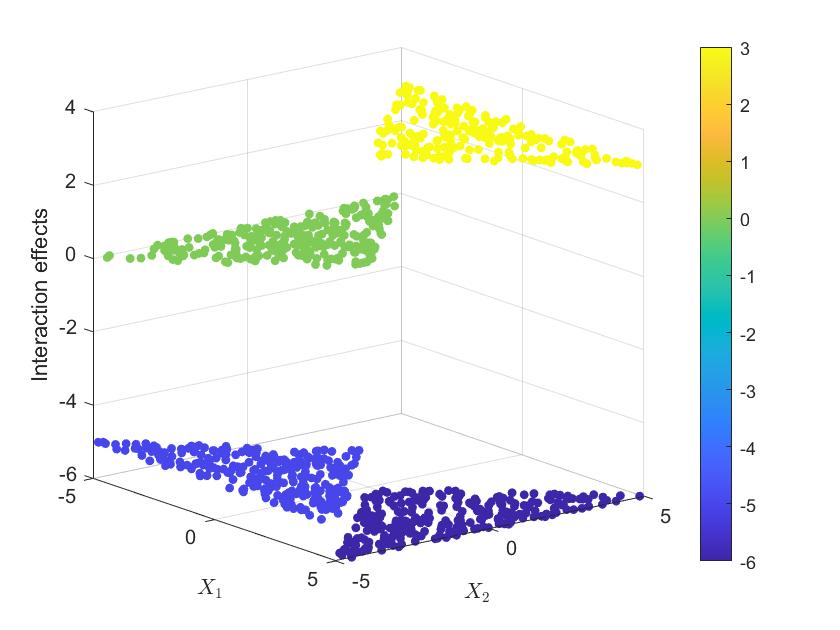}
    \caption{Interaction effects induced by the combination of two single treatments in different quadrants.}
    \label{fig: sim1_interaction_effects}
\end{figure}

\begin{table}[!ht]
    \centering
    \scriptsize
    \begin{tabular}{c|c|ccc}
    \hline
    Setting & Method & 400 & 800 & 2000\\
    \hline
    \multirow{7}{*}{1}&\textbf{MLRWL-Linear} & \textbf{4.104(0.092)} & \textbf{4.179(0.076)} & \textbf{4.238(0.077)} \\
    &\textbf{MLRWL-Kernel} & 4.022(0.094) & 4.156(0.073) & 4.226(0.074)\\
    &OWL-DL & 3.579(0.171) & 3.871(0.114) & 4.073(0.067) \\
    &$L_1$-PLS & 3.901(0.098) & 4.052(0.081) & 4.143(0.068) \\
    &OWL-MD & 3.772(0.119) & 3.944(0.108) & 4.035(0.110) \\
    &MOWL-Linear & 3.404(0.236) & 3.571(0.189) & 3.678(0.167) \\
    &MOWL-Kernel & 2.835(0.273) & 3.026(0.251) & 3.329(0.208) \\
    \hline
    \multirow{7}{*}{2}&\textbf{MLRWL-Linear} & 1.369(0.034) & 1.372(0.032) & 1.375(0.022) \\
    &\textbf{MLRWL-Kernel} & \textbf{1.700(0.047)} & \textbf{1.810(0.045)} & \textbf{1.923(0.046)} \\
    &OWL-DL & 1.451(0.060) & 1.472(0.051) & 1.499(0.042) \\
    &$L_1$-PLS & 1.371(0.062) & 1.364(0.051) & 1.377(0.044) \\
    &OWL-MD & 1.556(0.058) & 1.591(0.056) & 1.613(0.050) \\
    &MOWL-Linear & 1.553(0.074) & 1.578(0.054) & 1.589(0.037) \\
    &MOWL-Kernel & 1.641(0.063) & 1.668(0.041) & 1.681(0.028) \\
    \hline
    \multirow{7}{*}{3}&\textbf{MLRWL-Linear} & 5.664(0.426) & 6.000(0.418) & 6.267(0.247) \\
    &\textbf{MLRWL-Kernel} & \textbf{6.328(0.349)} & \textbf{6.415(0.100)} & \textbf{6.416(0.097)} \\
    &OWL-DL & 5.524(0.239) & 5.861(0.198) & 6.148(0.152) \\
    &$L_1$-PLS & 4.941(0.264) & 5.047(0.205) & 5.090(0.040) \\
    &OWL-MD &  5.924(0.189) & 6.125(0.152) & 6.295(0.110) \\
    &MOWL-Linear &  6.105(0.340) & 6.286(0.196) & 6.374(0.116) \\
    &MOWL-Kernel &  6.269(0.432) & 6.375(0.200) & 6.399(0.101) \\
    \hline
    \end{tabular}
    \caption{Simulation studies: mean and standard error of the value function under the proposed method with linear and nonlinear decision rules, and five competing methods: the outcome weighted learning with deep learning (OWL-DL, \citep{liang2018estimating}), the $L_1$ penalized least square ($L_1$-PLS, \citep{qian2011performance}), the outcome weighted learning with multinomial deviance (OWL-MD, \citep{huang2019multicategory}), and the multicategory outcome weighted learning with linear and kernel functions (MOWL-Linear and MOWL-Kernel, \citep{zhang2020multicategory}).}
    \label{tab: sim_results_value}
\end{table}

\begin{table}[!ht]
    \centering
    \scriptsize
    \begin{tabular}{c|c|ccc}
    \hline
    Setting & Method & 400 & 800 & 2000 \\
    \hline
    \multirow{7}{*}{1}&\textbf{MLRWL-Linear} & \textbf{0.797(0.033)} & \textbf{0.853(0.022)} & \textbf{0.884(0.012)} \\
    &\textbf{MLRWL-Kernel} & 0.664(0.042) & 0.744(0.023) & 0.797(0.015)  \\
    &OWL-DL & 0.534(0.037) & 0.581(0.029) & 0.625(0.022) \\
    &$L_1$-PLS & 0.669(0.019) & 0.699(0.015) & 0.717(0.012) \\
    &OWL-MD & 0.640(0.027) & 0.674(0.023) & 0.690(0.025)\\
    &MOWL-Linear & 0.552(0.051) & 0.591(0.039) & 0.611(0.034) \\
    &MOWL-Kernel & 0.421(0.045) & 0.456(0.042) & 0.509(0.036)\\
    \hline
    \multirow{7}{*}{2}&\textbf{MLRWL-Linear} & 0.262(0.015) & 0.267(0.011) & 0.272(0.009) \\
    &\textbf{MLRWL-Kernel} & \textbf{0.452(0.027)} & \textbf{0.539(0.025)} & \textbf{0.638(0.013)} \\
    &OWL-DL & 0.302(0.030) & 0.313(0.026) & 0.331(0.021)\\
    &$L_1$-PLS & 0.278(0.026) & 0.276(0.025) & 0.279(0.023) \\
    &OWL-MD & 0.353(0.027) & 0.365(0.020) & 0.377(0.017) \\
    &MOWL-Linear & 0.342(0.031) & 0.350(0.054) & 0.362(0.027) \\
    &MOWL-Kernel & 0.385(0.025) & 0.398(0.023) & 0.402(0.019) \\
    \hline
    \multirow{7}{*}{3}&\textbf{MLRWL-Linear} & 0.584(0.082) & 0.647(0.086) & 0.708(0.053) \\
    &\textbf{MLRWL-Kernel} & \textbf{0.717(0.114)} & \textbf{0.745(0.003)} & \textbf{0.746(0.003)} \\
    &OWL-DL & 0.455(0.064) & 0.543(0.062) & 0.635(0.052)\\
    &$L_1$-PLS & 0.272(0.055) & 0.257(0.045) & 0.237(0.040)\\
    &OWL-MD & 0.584(0.063) & 0.648(0.052) & 0.718(0.025)\\
    &MOWL-Linear & 0.649(0.114) & 0.709(0.056) & 0.739(0.017)\\
    &MOWL-Kernel & 0.705(0.142) & 0.737(0.063) & 0.746(0.005)\\
    \hline
    \end{tabular}
    \caption{Simulation studies: mean and standard error of the accuracy under the proposed method with linear and nonlinear decision rules, and five competing methods: the outcome weighted learning with deep learning (OWL-DL, \citep{liang2018estimating}), the $L_1$ penalized least square ($L_1$-PLS, \citep{qian2011performance}), the outcome weighted learning with multinomial deviance (OWL-MD, \citep{huang2019multicategory}), and the multicategory outcome weighted learning with linear and kernel functions (MOWL-Linear and MOWL-Kernel, \citep{zhang2020multicategory}).}
    \label{tab: sim_results_accuracy}
\end{table}

We compare our method with several existing methods which estimate the optimal ITR for combination treatments or multicategory treatments: the outcome weighted learning with deep learning (OWL-DL, \citep{liang2018estimating}), the $L_1$ penalized least square ($L_1$-PLS, \citep{qian2011performance}), the outcome weighted learning with multinomial deviance (OWL-MD, \citep{huang2019multicategory}), and the multicategory outcome weighted learning with linear and kernel functions (MOWL-Linear and MOWL-Kernel, \citep{zhang2020multicategory}). For the last four competing methods, we first convert the combination treatments into categorical treatments and apply those methods to estimate the ITR. 

All of the above simulation experiments are repeated 100 times, and the empirical performance is evaluated by the prediction accuracy, which is defined as $\frac{1}{n}\sum_{i=1}^{n}\mathbb{I}(\mathbf{A}_i^{\text{opt}} = \hat{d}(\mathbf{X}_i))$, where $\mathbf{A}_i^{\text{opt}}$ is the optimal treatment assignment for the $i$th subject derived from the data generation process. We also compute the empirical value function \citep{qian2011performance}, measured by an estimator $\mathcal{V}(\hat{d}) = \frac{\sum_{i=1}^{n}Y_i\mathbb{I}(\mathbf{A}_i = \hat{d}(\mathbf{X}_i))}{\sum_{i=1}^{n}\mathbb{I}(\mathbf{A}_i = \hat{d}(\mathbf{X}_i))}$ to assess the performance even when the optimal treatment assignments are unknown. 

The results of the simulation studies are presented in Tables \ref{tab: sim_results_value} and \ref{tab: sim_results_accuracy}, which demonstrates the effectiveness of our proposed method in estimating the optimal ITR for combination treatments. Our approach consistently outperformed competing methods in terms of optimal treatment assignment accuracy and value function across various settings, particularly considering the interaction effects of combination treatments. In simulation setting 1, the proposed method with linear decision rules improves the optimal treatment assignment accuracy by 13.2\% to 34.0\% compared with competing methods given 2000 samples. Even though the true decision rules are linear, the proposed method with nonlinear decision rules still achieves comparable accuracy and value function. In simulation settings 2 and 3, given that the true decision rules are nonlinear, the proposed method with nonlinear decision rules achieves the best performance. In particular, it improves the optimal treatment assignment accuracy by 27.1\% to 39.4\% in setting 2, and 4.1\% to 53.0\% in setting 3, respectively.

\section{Real Data Application}

In this section, we apply our method to recommend the optimal combination treatments for type-2 diabetes patients. The dataset is from Electronic Health Record (EHR) data accessible from the Clinical Practice Research Datalink \footnote{\url{https://cprd.com}}. In this study, type-2 diabetes patients were recruited from 2015 to 2018. Each subject was followed for 6 months, and the effectiveness of their assigned treatments was measured. There are four candidates for single treatments: dipeptidyl peptidase-4 (DPP4), sulfonylurea (SU), metformin (Met), and thiazolidinedione (TZD), which induces 16 combination treatments in total. In the past decade, researchers have investigated the combination treatments for type-2 diabetes patients \citep{salvo2016addition, ahren2008novel, mikhail2008combination}, and interaction effects among these treatments are evaluated. For example, \cite{salvo2016addition} suggest that SU combined with DDP4 induces a higher risk of hypoglycemia compared with using the SU treatment alone. Therefore, it is essential to consider interaction effects in assessing the optimal ITR for type-2 diabetes patients. 

In this dataset, 21 pre-treatment covariates were collected, including subjects' demographic information (e.g., age, BMI, gender, weight, height), diabetes-related health index (e.g., high-density lipoprotein, low-density lipoprotein, hematocrit), and medical history (e.g., congestive heart failure, stroke, hypertension). We use all these covariates except for the lower extremity arteries (LEA) to control for potential confounding, as the LEA value is the same for all subjects. The primary index to measure the effectiveness of treatment is the A1C, which measures average blood glucose levels \citep{kahn2008translating}. The normal A1C level is below 5.7\%, and type-2 diabetes patients are generally above 6.5\% \citep{zhang2010a1c}. The A1C levels are expected to decrease after the treatments are applied. Therefore, we use the negative change of A1C as our outcome, where a larger value indicates a better treatment effect.

In the implementation, we split the dataset into training (800), validation (200), and testing (139) sets. Since some combination treatments were assigned to fewer than 10 subjects, we perform stratified sampling to ensure that the training set includes all possible combination treatments. To validate the results, we repeat the sampling procedure and run the experiment independently 100 times, and report the averaged value function on test sets. Similar to the simulation studies, we compare the proposed method with the five competing methods which are used as competing methods in the simulation studies. 

Table \ref{tab: real_data_value} provides the means and standard errors of the value function. Our data analysis indicates that the proposed method under linear and nonlinear decision rules outperforms the competing methods with higher value functions and smaller standard deviations. Specifically, compared with the methods for the multicategory treatment ITR (all competing methods except for OWL-DL), the proposed method improves the value function by 48.6\%, 12.6\%, 16.9\%, and 3.8\%, while reducing the standard errors by 81.2\%, 51.8\%, 57.3\%, and 72.0\%, respectively. The improvement is partially due to the proposed MLRWL framework which requires estimating fewer decision rules than those multicategory ITR estimation methods. Therefore, the reduced standard error is also observed for OWL-DL \citep{liang2018estimating}. Compared with OWL-DL \citep{liang2018estimating}, our proposed method improves the value function by 4.4\% with a 29.7\% decreased standard errors. This suggests that incorporating interaction effects in estimating the optimal ITR for combination treatments is essential and useful. 

\begin{table}
    \centering
    \begin{tabular}{|c|c|}
    \hline
    Method & Value \\
    \hline
    \textbf{MLRWL-Linear} & 2.615(0.403) \\
    \textbf{MLRWL-Kernel} & \textbf{2.645(0.339)} \\
    OWL-DL & 2.534(0.482) \\
    $L_1$-PLS & 1.780(1.801) \\
    OWL-MD & 2.349(0.704) \\
    MOWL-Linear & 2.263(0.793)\\
    MOWL-Kernel & 2.548(1.210) \\
    \hline
    \end{tabular}
    \caption{Real data application: mean and standard error of the value function using the proposed method with linear and nonlinear decision rules, and five competing methods: the outcome weighted learning with deep learning (OWL-DL, \citealp{liang2018estimating}), the $L_1$ penalized least square ($L_1$-PLS, \citealp{qian2011performance}), the outcome weighted learning with multinomial deviance (OWL-MD, \citealp{huang2019multicategory}), the multicategory outcome weighted learning with linear and kernel functions (MOWL-Linear and MOWL-Kernel, \citealp{zhang2020multicategory}).}
    \label{tab: real_data_value}
\end{table}

\section{Conclusion and Discussion}

In this paper, we investigate the efficient estimation of individualized treatment rule for combination treatments. Our main contributions are as follows: First, we formulate the value maximization problem as a multi-label classification problem, which greatly reduces the modeling complexity for decision rules. Second, we proposed a non-convex multi-label surrogate loss which can incorporate any interaction effects among combination treatments. The proposed method has sound theoretical properties including Fisher consistency and universal consistency. Third, we solve the non-convex minimization efficiently with the difference-of-convex algorithm, and achieve great numerical performance in simulation studies and a real data example.

In the combination treatment problems, the positivity assumption is a contingent assumption, especially in observational study scenarios. We can explore further potential directions as follows. First, we could utilize parametric assumptions on the interaction effects among combination treatments. Suppose the high-order interaction effects do not exist, then it is possible to identify the treatment effects of combination treatment by single treatments and lower-order combination treatments. We refer readers to \citep{yu2023balancing} for a more comprehensive investigation in this direction. The second plausible solution is to identify the value functions with incremental propensity scores \citep{kennedy2019nonparametric, zhao2023positivity}, which shift the propensity values as a treatment assignment probability instead of assigning a deterministic treatment. This stochastic approach inherently avoids the positivity assumption; however, existing methods only apply to binary treatment problems. Therefore, it is worth further investigation on applying incremental propensity scores to multiple or combination treatment problems. Another potential solution is based on the pessimistic principal \citep{jin2022policy} which optimizes lower confidence bounds, instead of maximizing the point estimation of policy values. This approach can also relax the positivity assumption, but has not been studied in the combination treatment literature.


\begin{appendix}

\section{}

In this appendix, we provide the detailed derivation of the optimization problem for linear and nonlinear decision rules, and technique proof details of the theoretical properties of the Multi-Label Residual Weighted Learning (MLRWL). In addition, the extension of our method to observational study is also discussed.

\subsection{Derivation of the optimization problem of linear decision rules}
\label{A.1}

Within each iteration, the subproblem can be formulated as the following quadratic programming:

\begin{align}
\label{linear_primal_problem}
\scriptsize
\begin{split}
    \min_{\boldsymbol\beta} & \quad \frac{1}{2}\sum_{k=1}^{K}\lVert \boldsymbol\beta_{1k} \rVert^2 + \gamma\sum_{i=1}^{n}|w_i|\eta_i + \gamma\sum_{k=1}^{K}<\nabla_{\boldsymbol\beta_{1k}}\mathcal{L}_{\text{cave}}(\hat{\boldsymbol\beta}^{(t-1)}), \boldsymbol\beta_{1k}> + \\ &\gamma\sum_{k=1}^{K}\nabla_{\beta_{0k}}\mathcal{L}_{\text{cave}}(\hat{\boldsymbol\beta}^{(t-1)})\beta_{0k} \\
    s.t. & \quad \eta_i \ge I(w_i\ge 0) - a_{i}^{(k)}(\beta_{0k} + \mathbf{x}_i^T\boldsymbol\beta_{1k}), \text{ for any } k=1,2,...,K  \\
    & \quad \eta_i \ge 0 
\end{split}
\end{align}

where $\gamma$ is associated with the penalty coefficient $\lambda$. By introducing the Lagrange multipliers $\theta_{ik}$'s and $\mu_{i}$'s, we have the following Lagrange function: 

\begin{align}
\scriptsize
\begin{split}
    \mathcal{L}(\boldsymbol\beta, \boldsymbol\lambda, \mu) &= \frac{1}{2}\sum_{k=1}^{K}\lVert \boldsymbol\beta_{1k} \rVert^2 + \gamma \sum_{i=1}^{n}|w_i|\eta_i + \gamma\sum_{k=1}^{K}<\nabla_{\boldsymbol\beta_{1k}} L_{\text{cave}}(\hat{\boldsymbol\beta}^{(t-1)}), \boldsymbol\beta_{1k}> + \\
    &\gamma\sum_{k=1}^{K}\nabla_{\beta_{0k}}\mathcal{L}_{\text{cave}}(\hat{\boldsymbol\beta}^{(t-1)})\beta_{0k} \notag
    + \sum_{i=1}^{n}\sum_{k=1}^{K}\theta_{ik}(I(w_i\ge 0) - a_{i}^{(k)}(\beta_{0k} + \mathbf{x}_i^T\boldsymbol\beta_{1k}) - \eta_i)\\
    &- \sum_{i=1}^{n}\mu_i\eta_i, \notag 
\end{split}
\end{align}
where $\theta_{ik} \ge 0 \quad \forall i=1,...,n, k=1,...K$ and $\mu_{i} \ge 0, \forall i=1,...,n$. After taking derivatives of $\mathcal{L}(\boldsymbol\beta, \boldsymbol\theta, \mu)$ with respect to $\beta_{0k}$'s, $\boldsymbol\beta_{1k}$'s, and $\eta_i$'s and letting them equal to zero,  we have
\begin{equation}
\label{beta1_equation}
    \frac{\partial \mathcal{L}}{\partial \boldsymbol\beta_{1k}} = \boldsymbol\beta_{1k} + \gamma\nabla_{\boldsymbol\beta_{1k}}\mathcal{L}_{\text{cave}}(\hat{\boldsymbol\beta}^{(t-1)}) - \sum_{i=1}^{n}\theta_{ik}a_{i}^{(k)}\mathbf{x}_i = 0
\end{equation}
\begin{equation}
\label{beta0_equation}
    \frac{\partial \mathcal{L}}{\partial \beta_{0k}} = \gamma\nabla_{\beta_{0k}}\mathcal{L}_{\text{cave}}(\hat{\boldsymbol\beta}^{(t-1)}) - \sum_{i=1}^{n}\theta_{ik}a_{i}^{(k)} = 0
\end{equation}
\begin{equation}
\label{lambda_inequality}
    \frac{\partial \mathcal{L}}{\partial \eta_i} = \gamma |w_i| - \sum_{k=1}^{K}\theta_{ik} - \mu_i = 0.
\end{equation}
Then the primal problem \ref{linear_primal_problem} can be transformed to the dual problem:
\begin{align}
\label{linear_dual_problem}
\tiny
\begin{split}
    \min_{\theta} &\quad \frac{1}{2}\sum_{k=1}^{K}\sum_{i=1}^{n}\sum_{j=1}^{n}\theta_{ik}\theta_{jk}a_{i}^{(k)}a_{j}^{(k)}\mathbf{x}_i^T\mathbf{x}_j - \gamma\sum_{k=1}^{K}\sum_{i=1}^{n}\theta_{ik}a_{i}^{(k)}\mathbf{x}_i^T\nabla_{\boldsymbol\beta_{1k}}\mathcal{L}_{\text{cave}}(\hat{\boldsymbol\beta}^{(t-1)}) -\sum_{k=1}^{K}\sum_{i=1}^{n}\theta_{ik}I(w_i\ge 0) \\
    s.t. &\quad \sum_{k=1}^{K}\theta_{ik} \stackrel{(\ref{lambda_inequality})}{\le} \gamma |w_i|, \quad \gamma\nabla_{\beta_{0k}}\mathcal{L}_{\text{cave}}(\hat{\boldsymbol\beta}^{(t-1)}) \stackrel{(\ref{beta0_equation})}{=} \sum_{i=1}^{n}\theta_{ik}a_{i}^{(k)}, \quad \theta_{ik}\ge 0.
\end{split}
\end{align}
$\theta_{ik}$'s can be solved via the standard quadratic programming algorithm, and $\boldsymbol\beta_{1k}$ can be obtained from (\ref{beta1_equation}). By the Karush-Kuhn-Tucker conditions \citep{bertsekas1997nonlinear}, we have
\begin{align}
\begin{split}
    &\theta_{ik}(I(w_i \ge 0) - a_{i}^{(k)}(\beta_{0k} + \mathbf{x}_i^T\boldsymbol{\beta}_{1k}) - \eta_i) = 0 \notag \\
    &\mu_i\eta_i = 0. \notag
\end{split}
\end{align}
Then $\beta_{0k} = I(w_i \ge 0)a_i^{(k)} - \mathbf{x}_i^T\boldsymbol{\beta}_{1k}$ for points satisfying $\theta_{ik} > 0$ and $\eta_i = 0$. For numerical stability, we take the mean value of such $\beta_{0k}$'s as the estimation \citep{hastie2009elements}.

\subsection{Derivation of the optimization problem of nonlinear decision rules}
\label{A.2}

Similar to the linear case, we can still decompose the loss function into convex and concave parts, but replace the linear decision rule with a nonlinear decision rule, represented as $f_{k}(\mathbf{x}) = \beta_{0k} + \sum_{i=1}^{n}\mathcal{K}(\mathbf{x}_i, \mathbf{x})\beta_{ik}$ where $\mathcal{K}(\cdot, \cdot)$ is the pre-specified kernel function. Within the $t$ th iteration, we solve the following quadratic programming:
\begin{align}
\label{nonlinear_primal_problem}
\scriptsize
\begin{split}
    \min_{\beta} &\quad \frac{1}{2}\sum_{k=1}^{K} \boldsymbol{\beta}_{k}^T\mathbf{K}\boldsymbol{\beta_{k}} + \gamma\sum_{i=1}^{n}|w_i|\eta_i + \gamma\sum_{k=1}^{K}<\nabla_{\boldsymbol{\beta}_k}\mathcal{L}_{\text{cave}}(\hat{\boldsymbol\beta}^{(t-1)}), \boldsymbol{\beta}_k> + \\
    &\gamma\sum_{k=1}^{K}<\nabla_{\beta_{0k}}\mathcal{L}_{\text{cave}}(\hat{\boldsymbol\beta}^{(t-1)}), \beta_{0k}> \\
    s.t. & \quad \eta_i \ge I(w_i\ge 0) - a_{i}^{(k)}(\mathbf{K}_i\boldsymbol{\beta}_k + \beta_{0k}), \forall k=1,2,...,K  \\
    & \quad \eta_i \ge 0
\end{split}
\end{align}
where $\mathbf{K} = (K_{ij})_{n\times n}$ and $K_{ij} = \mathcal{K}(\mathbf{x}_i, \mathbf{x}_j)$, and $\mathbf{K}_i$ is the $i$ th row of $\mathbf{K}$. Following the similar procedure as in (\ref{beta1_equation}, \ref{beta0_equation}, \ref{lambda_inequality}), we can obtain the following subproblem in the $t$ th iteration:
\begin{align}
\label{nonlinear_dual_problem}
\scriptsize
\begin{split}
    \min_{\theta} & \quad\frac{1}{2}\sum_{k=1}^{K}\sum_{i=1}^{n}\sum_{j=1}^{n}\theta_{ik}\theta_{jk}a_{i}^{(k)}a_{j}^{(k)}\mathbf{K}_{ij} - \gamma\sum_{k=1}^{K}\sum_{i=1}^{n}\theta_{ik}a_{i}^{(k)}\nabla_{\beta_{ik}} \mathcal{L}_{\text{cave}}(\hat{\beta}^{(t-1)}) - \sum_{k=1}^{K}\sum_{i=1}^{n}\theta_{ik}I(w_i \ge 0) \\
    s.t. &\quad \sum_{k=1}^{K}\theta_{ik} \stackrel{(\ref{lambda_inequality})}{\le} \gamma |w_i|, \quad \gamma\nabla_{\beta_{0k}}\mathcal{L}_{\text{cave}}(\hat{\boldsymbol{\beta}}^{(t-1)}) = \sum_{i=1}^{n}\theta_{ik}a_i^{(k)}, \quad \theta_{ik}\ge 0.
\end{split}
\end{align}
Therefore, we can also apply the standard quadratic programming algorithm to solve (\ref{nonlinear_dual_problem}) and obtain the solution of $\theta_{ik}$'s.

\subsection{Algorithm Convergence}
\label{A.3}

In this section, we show that the convergent points of the Algorithm \ref{dc_alg} is stationary points.

\begin{proposition}
\label{prop: alg_convergence}
If the level set $\{\boldsymbol\beta|\mathcal{L}(\boldsymbol\beta) \le \mathcal{L}(\boldsymbol\beta^{(0)})\}$ is compact, then the convergent points obtained from Algorithm \ref{dc_alg} are stationary points of $\mathcal{L}(\boldsymbol\beta)$.
\end{proposition}

The level set condition for $\boldsymbol\beta^{(0)}$ is a standard assumption in the convergence analysis of non-convex programming \citep{khamaru2018convergence}. Note that Proposition \ref{prop: alg_convergence} does not exclude the possibility of local optima and saddle points, so the global optimum is not guaranteed. In practice, we can try multiple random initializations and select the ones that achieve the best performance on our validation sets.\\

\textit{Proof}: First of all, since $\boldsymbol\beta^{(t)} = \argmin_{\boldsymbol\beta}\mathcal{L}_{\text{cvx}}(\boldsymbol\beta) + <\nabla_{\boldsymbol{\beta}}\mathcal{L}_{\text{cave}}(\boldsymbol\beta^{(t-1)}), \boldsymbol\beta>$, it follows that $$\mathcal{L}_{\text{cvx}}(\boldsymbol\beta^{(t-1)}) + \nabla\mathcal{L}_{\text{cave}}(\boldsymbol\beta^{(t-1)})^T\boldsymbol\beta^{(t-1)} \ge \mathcal{L}_{\text{cvx}}(\boldsymbol\beta^{(t)}) + \nabla_{\boldsymbol{\beta}}\mathcal{L}_{\text{cave}}(\boldsymbol\beta^{(t-1)})^T\boldsymbol\beta^{(t)}.$$ After rearranging this inequality, we have
\begin{align}
    \mathcal{L}_{\text{cvx}}(\boldsymbol\beta^{(t-1)}) - \mathcal{L}_{\text{cvx}}(\boldsymbol\beta^{(t)}) \ge \nabla_{\boldsymbol{\beta}}\mathcal{L}_{\text{cave}}(\boldsymbol\beta^{(t-1)})^T(\boldsymbol\beta^{(t)} - \boldsymbol\beta^{(t-1)}). \notag
\end{align}
By the definition of (sub)gradient $\nabla_{\boldsymbol{\beta}}\mathcal{L}_{\text{cave}}(\boldsymbol\beta^{(t-1)})$, we have 
\begin{align}
    \mathcal{L}_{\text{cave}}(\boldsymbol\beta^{(t)}) \le \mathcal{L}_{\text{cave}}(\boldsymbol\beta^{(t-1)}) + \nabla_{\boldsymbol{\beta}}\mathcal{L}_{\text{cave}}(\boldsymbol\beta^{(t-1)})^T(\boldsymbol\beta^{(t)} - \boldsymbol\beta^{(t-1)}). \notag
\end{align}
Based on the above two inequalities, we can derive
\begin{align}
    \mathcal{L}_{\text{cvx}}(\boldsymbol\beta^{(t-1)}) + \mathcal{L}_{\text{cave}}(\boldsymbol\beta^{(t-1)})  \ge \mathcal{L}_{\text{cvx}}(\boldsymbol\beta^{(t)}) + \mathcal{L}_{\text{cave}}(\boldsymbol\beta^{(t)}), \notag
\end{align}
which indicates that the sequence $\{\mathcal{L}(\boldsymbol\beta^{(t)})\}$ is monotonically decreasing. 

Under the assumption that the initial values $\boldsymbol\beta^{(0)}$ has the following property: the level set $\{\boldsymbol\beta | \mathcal{L}(\boldsymbol\beta) \le \mathcal{L}(\boldsymbol\beta^{(0)})\}$ is compact, then the sequence $\{\boldsymbol\beta^{(t)}\}$ has a limit point $\boldsymbol\beta^{*}$ by the Bolzano-Weierstrass theorem \citep{bartle2000introduction}. 

Next, we prove that $\boldsymbol\beta^{*}$ is a stationary point. Due to the convexity of $\mathcal{L}_{\text{cvx}}(\boldsymbol\beta)$ and $-\mathcal{L}_{\text{cave}}(\boldsymbol\beta)$, the (sub)gradients exist. Furthermore, we have $\nabla\mathcal{L}_{\text{cvx}}(\boldsymbol\beta^{(t)}) + \nabla\mathcal{L}_{\text{cave}}(\boldsymbol\beta^{(t-1)}) = 0$ and $\nabla\mathcal{L}_{\text{cave}}(\boldsymbol\beta^{(t-1)})$ converges to $-\nabla\mathcal{L}_{\text{cvx}}(\boldsymbol\beta^{*})$. Thus, the limit point $\boldsymbol\beta^{*}$ is a stationary point since $\nabla\mathcal{L}(\boldsymbol\beta^{*}) = \nabla\mathcal{L}_{\text{cvx}}(\boldsymbol\beta^{*}) + \nabla\mathcal{L}_{\text{cave}}(\boldsymbol\beta^{*}) = 0$.

\subsection{Proof of Lemma \ref{lemma: fisher_consistency}}
\label{A.4}

First, we show the Fisher consistency of the proposed method under the outcome-weighted framework, i.e., the weight in the risk is $\frac{Y}{\mathbb{P}(\mathbf{A}|\mathbf{X})}$, and the associated risk is as follows
\begin{align}
    \mathcal{R}_{\psi}(f) &= \mathbb{E}[\frac{Y}{\mathbb{P}(\mathbf{A}|\mathbf{X})}\psi(\mathbf{A}, f(\mathbf{X}))] \notag \\
    &= \mathbb{E}[\frac{Y}{\mathbb{P}(\mathbf{A}|\mathbf{X})}(T_1(\mathbf{A}, f(\mathbf{X})) - T_0(\mathbf{A}, f(\mathbf{X})))]. \notag
\end{align}
For any $\mathbf{X}=\mathbf{x}$, the conditional risk is 
\begin{align}
    &\mathbb{E}[\frac{Y}{\mathbb{P}(\mathbf{A}|\mathbf{X})}(T_1(\mathbf{A}, f(\mathbf{X})) - T_0(\mathbf{A}, f(\mathbf{X})))|\mathbf{X} = \mathbf{x}] \notag \\
    =&\sum_{\mathbf{a}\in \mathcal{A}}\mathbb{E}[Y(T_1(\mathbf{A}, f(\mathbf{X})) - T_0(\mathbf{A}, f(\mathbf{X})))|\mathbf{X}=\mathbf{x}, \mathbf{A}=\mathbf{a}]. \notag
\end{align}
Note that for any measurable functions $f(\mathbf{x}) = (f^{(1)}(\mathbf{x}), f^{(2)}(\mathbf{x}), ..., f^{(K)}(\mathbf{x}))$, there exists only one $\mathbf{a}\in\mathcal{A}=\{-1, 1\}^{K}$ (denoted as $\mathbf{a}_*$) such that $a_{*}^{(k)}f^{(k)}(\mathbf{x}) \ge 0$ for all $k\in\{1, 2, ..., K\}$. For any other $\mathbf{a}\neq \mathbf{a}^{*}$, there exists $k_0 \in \{1, 2, ..., K\}$ such that $a^{(k_0)}f^{(k_0)}(\mathbf{x}) \le 0$, then $\min_{k}a^{(k)}f^{(k)}(\mathbf{x}) \le 0$. And if we denote $k_1 = \argmin a^{(k)}f^{(k)}(\mathbf{x})$, we can obtain $T_1(\mathbf{a}, f(\mathbf{x})) = 1 - a^{(k_1)}f^{(k_1)}(\mathbf{x})$ and $T_0(\mathbf{a}, f(\mathbf{x})) = - a^{(k_1)}f^{(k_1)}(\mathbf{x})$, which yields $T_1(a, f(\mathbf{x})) - T_0(a, f(\mathbf{x})) = 1$.  Following the above derivation, we have
\begin{align}
    \scriptsize
    \begin{split}
    &\sum_{\mathbf{a}\in \mathcal{A}}\mathbb{E}[Y|\mathbf{X}=\mathbf{x}, \mathbf{A}=\mathbf{a}](T_1(\mathbf{a}, f(\mathbf{x})) - T_0(\mathbf{a}, f(\mathbf{x}))) \notag \\
    =& \sum_{\mathbf{a}\in\mathcal{A}}\mathbb{E}[Y|\mathbf{X}=\mathbf{x}, \mathbf{A}=\mathbf{a}] + E[Y|\mathbf{X}=\mathbf{x}, \mathbf{A}=\mathbf{a}_*](-1 + T_1(\mathbf{a}_*, f(\mathbf{x})) - T_0(\mathbf{a}_*, f(\mathbf{x}))). \notag
    \end{split}
\end{align}
Note that $-1 + T_1(\mathbf{a}_*, f(\mathbf{x})) - T_0(\mathbf{a}_*, f(\mathbf{x})) \le 0$, then we have $\mathbb{E}[Y|\mathbf{X}=\mathbf{x}, \mathbf{A}=\mathbf{a}_*](-1 + T_1(\mathbf{a}_*, f(\mathbf{x})) - T_0(\mathbf{a}_*, f(\mathbf{x}))) \ge 0$ for any measurable $f(\mathbf{x})$ if $\mathbb{E}[Y|\mathbf{X}=\mathbf{x}, \mathbf{A}=\mathbf{a}_*] < 0$. Meanwhile, $\mathbb{E}[Y|\mathbf{X}=\mathbf{x}, \mathbf{A}=\mathbf{a}_*](-1 + T_1(\mathbf{a}_*, f(\mathbf{x})) - T_0(\mathbf{a}_*, f(\mathbf{x})))$ $\le 0$ for any measurable $f(\mathbf{x})$ if $\mathbb{E}[Y|\mathbf{X}=\mathbf{x}, \mathbf{A}=\mathbf{a}_*] > 0$. Hence, the conditional risk is minimized when $a_* = \argmax_{a\in\mathcal{A}}\mathbb{E}[Y|\mathbf{X}=\mathbf{x}, \mathbf{A}=\mathbf{a}]$ and $a_*^{(k)}f^{(k)}(\mathbf{x})\ge 1$ for any $k\in \{1, 2, ..., K\}$. In other words, the minimizer $\hat{f}(\cdot)$ of $\mathcal{R}_{\psi}(f)$ satisfies $d(\mathbf{x}) = \text{sign}(\hat{f}(\mathbf{x})) = \argmax_{\mathbf{a}\in\mathcal{A}}\mathbb{E}[Y|\mathbf{X} = \mathbf{x}, \mathbf{A} = \mathbf{a}]$. \qed

\subsection{Proof of Theorem \ref{thm: fisher_consistency}}
\label{A.5}

Following the steps in \ref{A.4},  if we adopt the residual $Y - g(\mathbf{X})$ as the weight, then we have similar conclusions as in Section S.4:
\begin{align}
    &\sum_{\mathbf{a}\in \mathcal{A}}\mathbb{E}[Y-g(\mathbf{X})|\mathbf{X}=\mathbf{x}, \mathbf{A}=\mathbf{a}](T_1(\mathbf{a}, f(\mathbf{x})) - T_0(\mathbf{a}, f(\mathbf{x}))) \notag \\
    =& \sum_{\mathbf{a}\in\mathcal{A}}\mathbb{E}[Y-g(\mathbf{X})|\mathbf{X}=\mathbf{x}, \mathbf{A}=\mathbf{a}] \notag \\
    + &\mathbb{E}[Y-g(\mathbf{X})|\mathbf{X}=\mathbf{x}, \mathbf{A}=\mathbf{a}_*](-1 + T_1(\mathbf{a}_*, f(\mathbf{x})) - T_0(\mathbf{a}_*, f(\mathbf{x}))).
\end{align}
Since $\argmax_{a\in\mathcal{A}}\mathbb{E}[Y-g(\mathbf{X})|\mathbf{X}=\mathbf{x}, \mathbf{A}=\mathbf{a}] = \argmax_{\mathbf{a}\in\mathcal{A}}\mathbb{E}[Y|\mathbf{X}=\mathbf{x}, \mathbf{A}=\mathbf{a}]$, the desired results are concluded. \qed

\subsection{Proof of Theorem \ref{thm: excess_risk}}
\label{A.6}

In this proof, we will follow two steps to prove the results. In the first step, we first introduce intermediate risks $\mathcal{R}_g(f)$ and $\mathcal{R}^*_g$ and build connection between $\mathcal{R}_g(f) - \mathcal{R}^*_g$ and $\mathcal{R}_{\psi, g}(f) - \mathcal{R}_{\psi, g}^{*}$. In the second step, we establish the equivalence between $\mathcal{R}_g(f) - \mathcal{R}^*_g$ and $\mathcal{R}(f) - \mathcal{R}^*$ and conclude the results.

Now, we introduce an intermediate risk given $g(\mathbf{X})$ and 0-1 loss: 
\begin{align}
    \mathcal{R}^*_{g} &= \mathbb{E}\bigg[\sum_{a\in\mathcal{A}}\mathbb{E}[Y-g(\mathbf{X})|\mathbf{X}=x, \mathbf{A}=a]\mathbb{I}(a \neq \text{sign}(f_*(\mathbf{x})))\bigg], \notag \\
    \mathcal{R}_{g}(f) &= \mathbb{E}\bigg[\sum_{a\in\mathcal{A}}\mathbb{E}[Y-g(\mathbf{X})|\mathbf{X}=x, \mathbf{A}=a]\mathbb{I}(a \neq \text{sign}(f(\mathbf{x})))\bigg], \notag
\end{align}

From the proof of Lemma \ref{lemma: fisher_consistency} and Theorem \ref{thm: fisher_consistency}, we have 
\begin{align}
    \mathcal{R}_{\psi, g}(f) &= \mathbb{E}\bigg[\sum_{a\in\mathcal{A}}\mathbb{E}[Y-g(\mathbf{X})|\mathbf{X} = x, \mathbf{A} = a]\big\{T_1(a, f(x)) - T_0(a, f(x))\big\}\bigg] \notag \\
    \mathcal{R}_{\psi, g}^{*} &= \mathbb{E}\bigg[\sum_{a\in\mathcal{A}}\mathbb{E}[Y-g(\mathbf{X})|\mathbf{X} = x, \mathbf{A} = a]\big\{T_1(a, f_{*}(x)) - T_0(a, f_{*}(x))\big\}\bigg]. \notag
\end{align}
In addition, $|f_*^{(k)}| \ge 1$, and $a^{(k)}_*f_*^{(k)}(x) \ge 1$ for $a_* =\argmax_{a\in\mathcal{A}}\mathbb{E}[Y-g(\mathbf{X})|\mathbf{X}=X, \mathbf{A}=a]$. Therefore, for any other $a\in\mathcal{A}$, there exists $k_0$ such that $a^{(k_0)}f^{(k_0)}_*(\mathbf{x})$ $\le -1$, which leads the generalized $\psi$-loss $T_1(a, f) - T_0(a, f) = 1$. Then we have 
\begin{align}
    \mathcal{R}_{\psi, g}^* = \mathbb{E}\bigg[\sum_{a\in\mathcal{A}\backslash a_*}\mathbb{E}[Y-g(\mathbf{X})|\mathbf{X}=x, \mathbf{A}=a]\bigg]. \notag
\end{align}

Similarly, we can find 
\begin{align}
    \mathcal{R}^*_{g} &= \mathbb{E}\bigg[\sum_{a\in\mathcal{A}\backslash a_*}\mathbb{E}[Y-g(\mathbf{X})|\mathbf{X}=x, \mathbf{A}=a]\bigg]. \notag
\end{align}

Therefore, it is sufficient to prove that $\mathcal{R}_{\psi, g}(f) \ge \mathcal{R}_{g}(\text{sign}(f))$ to establish the first excess risk bound. Note that for any $f$, there only exists one combination treatment $\mathbf{a}_*$ such that $a^{(k)}_*f^{(k)}_*(\mathbf{x}) > 0$. And for any other $\mathbf{a} \neq \mathbf{a}_*$, there exists $k_0$ such that $a^{(k_0)}f^{(k_0)}(\mathbf{x}) < 0$. Therefore, $\psi(\mathbf{a}, f(\mathbf{x})) = 1$ for any $\mathbf{a}\neq \mathbf{a}_*$, followed by 
\begin{align}
&\mathbb{E}\bigg[\sum_{\mathbf{a}\neq \mathbf{a}_*}\mathbb{E}[Y-g(\mathbf{X})|\mathbf{X} = \mathbf{x}, \mathbf{A} = \mathbf{a}]\psi(\mathbf{a}, f(\mathbf{x}))\bigg] \notag \\
= &\mathbb{E}\bigg[\sum_{\mathbf{a}\neq \mathbf{a}_*}\mathbb{E}[Y-g(\mathbf{X})|\mathbf{X} = \mathbf{x}, \mathbf{A} = \mathbf{a}]\mathbb{I}(\mathbf{a}\neq\text{sign}(f(\mathbf{x})))\bigg].  \notag
\end{align}
Since $\mathbb{E}[Y-g(\mathbf{X})|\mathbf{X}=\mathbf{x}, \mathbf{A}=\mathbf{a}_*] > 0$, and $\psi(\mathbf{a}, f(\mathbf{x})) \ge \mathbb{I}(\mathbf{a}\neq f(\mathbf{x}))$ for any measurable $f$, we conclude that $\mathcal{R}_{\psi, g}(f) \ge \mathcal{R}_{g}(\text{sign}(f))$.

Given that $\mathbb{E}[g(\mathbf{X})]$ is a constant, $\mathcal{R}_g(\text{sign}(f)) = \mathcal{R}(\text{sign}(f)) - \mathbb{E}[g(\mathbf{X})]$ for any measurable $f$, and $\mathcal{R}^*_g = \mathcal{R}^* - \mathbb{E}[g(\mathbf{X})]$, so $\mathcal{R}_g(\text{sign}(f)) - \mathcal{R}_g^* = \mathcal{R}(\text{sign}(f)) - \mathcal{R}^*$, which concludes the results. \qed

\subsection{Proof of Theorem \ref{thm: consistency}}
\label{A.7}

First, let $L(h, b) = \frac{Y - g(\mathbf{X})}{\mathbb{P}(\mathbf{A}|\mathbf{X})}\psi(Z^{(1)}, ..., Z^{(K)})$, where $Z^{(k)} = A^{(k)}(h^{(k)}(\mathbf{X}) + b^{(k)})$, $h^{(k)}(\cdot) \in \mathcal{H}_{\mathcal{K}}$ and $b^{(k)} \in \mathbb{R}$. For the minimizer of the empirical loss (\ref{empirical_loss}), we denote the corresponding estimator as $h_n$ and $b_n$, respectively. By the definition of $h_n$, $b_n$, for any $h^{(k)}\in\mathcal{H}_{\mathcal{K}}$ and $b^{(k)}\in\mathbb{R}$, we have
\begin{align}
    \mathbb{P}_{n}(L(h_n, b_n)) \le \mathbb{P}_{n}(L(h_n, b_n)) + \frac{\lambda}{2}\lVert h_n\rVert_{\mathcal{K}}^2 \le \mathbb{P}_n(L(h, b)) + \frac{\lambda}{2}\lVert h\rVert_{\mathcal{K}}^2, \notag
\end{align}
where $\mathbb{P}_n$ denotes the empirical measure of the observed datasets $(\mathbf{x}_i, \mathbf{a}_i, y_i)_{i=1}^{n}$. Then, $\lim\sup_n\mathbb{P}_n(L(h_n, b_n)) \le \mathbb{P}(L(h, b)) = \mathcal{R}_{\psi}(h + b)$ almost surely. Furthermore, it implies that 
\begin{align}
    \lim\sup_n\mathbb{P}_n(L(h_n, b_n)) \le \inf_{h^{(k)}\in\mathcal{H}_{\mathcal{K}}, b^{(k)}\in\mathbb{R}}\mathcal{R}_{\psi}(h + b) \le \mathbb{P}(L(h_n, b_n)),\quad\text{ w.p. 1.} \notag
\end{align}
Therefore, it is suffice to show that $\mathbb{P}_n(L(h_n, b_n)) - \mathbb{P}(L(h_n, b_n)) \rightarrow 0$ in probability to conclude the results.

In the following, we establish the bound for $\lVert h_n \rVert_{\mathcal{K}}$ and $b_n$ to control the complexity of the space $\mathcal{H}_{\mathcal{K}} + \{1\}$. Since $\mathbb{P}_n(L(h_n, b_n)) + \frac{\lambda}{2}\lVert h_n\rVert^2_{\mathcal{K}} \le \mathbb{P}_n(L(h, b)) + \frac{\lambda}{2}\lVert h\rVert_{\mathcal{K}}^2$ for any $h$ and $b$, we take $h = \mathbf{0}$ and $b = 0$, to obtain that 
\begin{align}
    \mathbb{P}_n(L(h_n, b_n)) + \frac{\lambda}{2}\lVert h_n\rVert^2_{\mathcal{K}} \le \mathbb{P}_n(\frac{Y - g(\mathbf{X})}{\mathbb{P}(\mathbf{A}|\mathbf{X})}). \notag
\end{align}
Note that $0 \le \psi(\cdot) \le 1$, we can derive
\begin{align}
    \lambda\lVert h_n\rVert^2_{\mathcal{K}} \le 4\mathbb{P}_n(\frac{|Y - g(\mathbf{X})|}{\mathbb{P}(\mathbf{A}|\mathbf{X})}) \le 4M. \notag
\end{align}

To obtain the bound for $b_n$, we note that there exists some $\mathbf{x}_i$ such that $|h_n(\mathbf{x}_i) + b_n| < 1$, then we have 
\begin{align}
    |b_n| \le 1 + |h_n(\mathbf{x}_i)| \le  1 + \lVert h_n \rVert_{\infty} \le 1 + C_{\mathcal{K}}\lVert h_n\rVert_{\mathcal{K}}. \notag
\end{align}
Therefore, we can obtain that $|\sqrt{\lambda}b_n| \le \sqrt{\lambda} + C_{\mathcal{K}}\sqrt{\lambda}\lVert h_n\rVert_{\mathcal{K}}$. Since $\lambda \rightarrow 0$, $C_{\mathcal{K}}$ and $\sqrt{\lambda}\lVert h_n\rVert_{\mathcal{K}}$ are bounded, $|\sqrt{\lambda}b_n|$ is bounded too. Furthermore, since $\psi(\cdot)$ is a Lipschitz continuous function with Lipschitz constant 1, the class $\{\sqrt{\lambda}L(h, b): \lVert\sqrt{\lambda}h\rVert_{\mathcal{K}}, |\sqrt{\lambda}b| \text{ are bounded}\}$ is a P-Donsker class, which induces
\begin{align}
    \sqrt{n\lambda}(\mathbb{P}_n(L(h_n, b_n)) - \mathbb{P}(L(h_n, b_n)))  = O_p(1). \notag
\end{align}
Consequently, as $n\lambda \rightarrow \infty$, we have $\mathbb{P}_n(L(h_n, b_n)) \rightarrow \mathbb{P}(L(h_n, b_n))$ in probability. \qed

\subsection{Estimation of Working Models of Treatment-free Effects and Propensity Score}
\label{A.8}

In observational studies, the treatment assignment is usually unknown to practioners. Therefore, it is essential to estimate the propensity score before estimating the ITR via (\ref{empirical_loss}). In this work, we utilize the multinomial logistic regression to estimate the propensity score. Specifically, we first encode the combination treatment with categorical codings $\tilde{A}_i$: $\{1, ..., 2^{K}\}$, and then maximize the likelihood:
\begin{align}
    \max_{\tau_1, ..., \tau_{2^{K}}}\sum_{i=1}^{n}\sum_{j=1}^{2^{K}}\mathbb{I}(\tilde{A}_i = j)\log\frac{\exp(\mathbf{X}_i^T\tau_j)}{\sum_{j}\exp(\mathbf{X}_i^T\tau_j)} - \lambda\sum_{j}\lVert\tau_j\rVert_2^2, \notag
\end{align}
and the estimated propensity score is $\mathbb{P}(\tilde{A}_i|\mathbf{X}_i) = \frac{\exp(\mathbf{X}_i^T\tau_{\tilde{A}_i})}{\sum_{j}\exp(\mathbf{X}_i^T\tau_j)}$.

As for the treatment-free effects, $g(\mathbf{X}) = \frac{1}{|\mathcal{A}|}\sum\mathbb{E}[Y|\mathbf{X}, \mathbf{A}] = \mathbb{E}[\frac{Y}{|\mathcal{A}|\mathbb{P}(\mathbf{A}|\mathbf{X})}|\mathbf{X}]$, so we assume as linear model to fit the treatment-free effects and obtain an estimation by minimizing the following loss:
\begin{align}
    \min_{\eta}\sum_{i=1}^{n}\frac{1}{\hat{\mathbb{P}}(\mathbf{A}|\mathbf{X})}(Y_i - \mathbf{X}_i^T\eta)^2. \notag
\end{align}
For clinical trials with uniform random assignment, the above loss reduces to 
\begin{align}
    \min_{\eta}\sum_{i=1}^{n}(Y_i - \mathbf{X}_i^T\eta)^2. \notag
\end{align}

\subsection{Consistency of $\hat{f}_n$ in Observational Study}
\label{A.9}

In this section, we show the consistency of the proposed method in observational study, where the propensity score model is also estimated from finite sample data. The following Theorem states the necessary assumptions and the consistency of the proposed estimator.

\begin{theorem}
\label{thm: consistency_os}
Suppose the penalty coefficient $\lambda$ in the primal form (\ref{loss_decomp}) satisfies $\lambda \rightarrow 0$ and $n\lambda \rightarrow \infty$. The weights $|Y - g(\mathbf{X})|/\mathbb{P}(\mathbf{A}|\mathbf{X})$'s are upper bounded by some positive constant $M$ almost surely. Suppose the working model of propensity score $\mathbb{P}(\mathbf{A}|\mathbf{X}; \tau_n)$ is a uniform consistent estimator of the true propensity score model, say, $\lVert \tau_n - \tau \rVert \rightarrow 0$ in probability and it is bounded below by some constant $\xi > 0$ for any $\mathbf{X}\in\mathcal{X}$ and $\mathbf{A}\in\mathcal{X}$. Then for any distribution $P$ for $(\mathbf{X}, \mathbf{A}, Y)$, we have
\begin{align}
    \mathbb{P}\bigg\{\lim_{n\rightarrow\infty}\mathcal{R}_{\psi, g}(\hat{f}_n) = \inf_{f\in\mathcal{H}_{\mathcal{K}} + \{1\}}\mathcal{R}_{\psi, g}(f)\bigg\} = 1, \notag
\end{align}
where $\hat{f}_n$ is the minimizer of the empirical loss (\ref{empirical_loss}) with plug-in estimator of propensity score $\hat{\mathbb{P}}(\mathbf{A}|\mathbf{X})$. $\mathcal{H}_{\mathcal{K}} + \{1\}$ denotes the shifted reproducing kernel Hilbert space we considered in Section \ref{sec: nldr}.
\end{theorem}

\begin{proof}
We first introduce some notations for the ease of derivation. First, we denote the propensity score model as $\mathbb{P}(\mathbf{A}|\mathbf{X}; \tau)$ where $\mathbb{P}(\cdot|\cdot)$ specifies the function form, and $\tau$ is the associated parameter. The estimated propensity score is denoted as $\mathbb{P}(\mathbf{A}|\mathbf{X}; \tau_n)$, where $\tau_n$ is the finite sample estimator of $\tau$. In addition, we define $L(h, b, \tau) = \frac{Y-g(\mathbf{X})}{\mathbb{P}(\mathbf{A}|\mathbf{X}; \tau)}\psi(Z^{(1)}, ..., Z^{(K)})$, and $L(h, b, \tau_n) = \frac{Y-g(\mathbf{X})}{\mathbb{P}(\mathbf{A}|\mathbf{X}; \tau_n)}\psi(Z^{(1)}, ..., Z^{(K)})$. Therefore, we have
\begin{align}
    h_n, b_n = \argmin_{h^{(k)}\in\mathcal{H}_{\mathcal{K}}, b^{(k)}\in\mathbb{R}}\mathbb{P}_n(L(h, b, \tau_n)). \notag
\end{align}

Our expected result can be expressed as 
\begin{align}
\lim_{n\rightarrow\infty}\mathbb{P}(L(h_n, b_n, \tau)) = \inf_{h^{(k)}\in \mathcal{H}_{\mathcal{K}}, b^{(k)}\in\mathbb{R}}\mathbb{P}(L(h, b, \tau)). \notag 
\end{align}

The ($\ge$) part is straightforward in that
\begin{align}
    \inf_{h^{(k)}\in \mathcal{H}_{\mathcal{K}}, b^{(k)}\in\mathbb{R}}&\mathbb{P}(L(h, b, \tau)) \le \mathbb{P}(L(h_n, b_n, \tau)), \notag
\end{align}
and it is followed by 
\begin{align}
\inf_{h^{(k)}\in \mathcal{H}_{\mathcal{K}}, b^{(k)}\in\mathbb{R}}\mathbb{P}(L(h, b, \tau)) \le \lim_{n\rightarrow \infty}\mathbb{P}(L(h_n, b_n, \tau)). \notag
\end{align}

For the ($\le$) part, we can decompose the difference as follows:
\begin{align}
    \mathbb{P}(L(h_n, b_n, \tau)) - \mathbb{P}(L(h, b, \tau)) &= \mathbb{P}(L(h_n, b_n, \tau)) - \mathbb{P}(L(h_n, b_n, \tau_n)) \notag \\
    &+ \mathbb{P}(L(h_n, b_n, \tau_n)) - \mathbb{P}_n(L(h_n, b_n, \tau_n)) \notag \\
    &+ \mathbb{P}_n(L(h_n, b_n, \tau_n)) - \mathbb{P}_n(L(h, b, \tau_n)) \notag \\
    &+ \mathbb{P}_n(L(h, b, \tau_n)) - \mathbb{P}_n(L(h, b, \tau)) \notag \\
    &+ \mathbb{P}_n(L(h, b, \tau)) - \mathbb{P}(L(h, b, \tau)) \notag \\
    &= (I) + (II) + (III) + (IV) + (V), \notag
\end{align}
where the term $(III)$ is negative by the definition of $h_n$ and $b_n$, and the term $(V)$ is easily goes to zero in probability based on weak law of large number. Therefore, we only need to consider the asymptotic properties of the terms $(I), (II), (IV)$.

For the term $(I)$, it is easy to see 
\begin{align}
    \mathbb{P}\bigg\{\frac{Y-g(\mathbf{X})}{\mathbb{P}(\mathbf{A}|\mathbf{X}; \tau)}[1 - \frac{\mathbb{P}(\mathbf{A}|\mathbf{X}; \tau)}{\mathbb{P}(\mathbf{A}|\mathbf{X}; \tau_n)}]\psi(h_n, b_n)\bigg\} \rightarrow 0, \notag
\end{align}
due to the boundedness of $\frac{Y-g(\mathbf{X})}{\mathbb{P}(\mathbf{A}|\mathbf{X}; \tau)}$ and $0 \le \psi(h_n, b_n) \le 1$.

For the term (II), we will use empirical process theory to prove this convergence. Before that, we establish the bound for $h_n$, $b_n$ and $\tau_n$ to control the complexity. By the same means, we have 
\begin{align}
    \mathbb{P}_n(L(h_n, b_n, \tau_n)) + \frac{\lambda}{2}\lVert h_n\rVert_{\mathcal{K}}^2 \le \mathbb{P}_n(L(h, b, \tau_n)) + \frac{\lambda}{2}\lVert h\rVert_{\mathcal{K}}^2, \notag
\end{align}
and we can take $h = \mathbf{0}$ and $b = 0$, so we have 
\begin{align}
    \lambda\lVert h_n\rVert_{\mathcal{K}}^2 &\le 4\mathbb{P}_n(\frac{|Y-g(\mathbf{X})|}{\mathbb{P}(\mathbf{A}|\mathbf{X}; \tau_n)}) \notag \\
    &\le 4\mathbb{P}_n(\frac{|Y-g(\mathbf{X})|}{\mathbb{P}(\mathbf{A}|\mathbf{X}; \tau)}\frac{\mathbb{P}(\mathbf{A}|\mathbf{X}; \tau_n)}{\mathbb{P}(\mathbf{A}|\mathbf{X}; \tau)}) \notag \\
    &\le 4M/p_{\mathcal{A}}. \notag
\end{align}
The bound for $b_n$ can be derived as the same approach as in Appendix \ref{A.7}, in that 
\begin{align}
    |\sqrt{\lambda}b_n| \le \sqrt{\lambda} + \mathcal{C}_{\mathcal{K}}\lVert h_n\rVert_{\mathcal{K}}. \notag
\end{align}
In summary, the class $\{\sqrt{\lambda}L(h, b, \tau): \lVert\sqrt{\lambda}h\rVert_{\mathcal{K}}, |\sqrt{\lambda}b|, \lVert \sqrt{\lambda}\tau\rVert_{2} \text{ are bounded}\}$ is a P-Donsker class, which induces 
\begin{align}
    \sqrt{n\lambda}(\mathbb{P}_n(L(h_n, b_n, \tau_n)) - \mathbb{P}(L(h_n, b_n, \tau_n))) = O_p(1). \notag
\end{align}
Consequently, as $n\lambda \rightarrow \infty$, we have $\mathbb{P}_n(L(h_n, b_n, \tau_n)) \rightarrow \mathbb{P}(L(h_n, b_n, \tau_n))$ in probability.

For term (IV), we first consider the upper bound of the difference
\begin{align}
    \left|L(h, b, \tau_n) - L(h, b, \tau)\right| &= \left|\frac{Y_i - g(\mathbf{X}_i)}{\mathbb{P}(\mathbf{A}_i| \mathbf{X}_i; \tau_n)} - \frac{Y_i - g(\mathbf{X}_i)}{\mathbb{P}(\mathbf{A}_i| \mathbf{X}_i; \tau)}\right|\psi(h, b) \notag \\
    &\le \left|\frac{Y_i-g(\mathbf{X}_i)}{\mathbb{P}(\mathbf{A}_i|\mathbf{X}_i; \tau)}(\frac{\mathbb{P}(\mathbf{A}_i|\mathbf{X}_i; \tau)}{\mathbb{P}(\mathbf{A}_i|\mathbf{X}_i; \tau_n)} - 1)\right| \notag \\
    &\le M\left|\frac{\mathbb{P}(\mathbf{A}_i|\mathbf{X}_i; \tau)}{\mathbb{P}(\mathbf{A}_i|\mathbf{X}_i; \tau_n)}-1\right| \notag
\end{align}
Since $\tau_n \rightarrow \tau$ uniformly, for any $\epsilon > 0$, there exists $N_{\epsilon}$ such that if $n > N_{\epsilon}$, $\left|\frac{\mathbb{P}(\mathbf{A}_i|\mathbf{X}_i; \tau)}{\mathbb{P}(\mathbf{A}_i|\mathbf{X}_i; \tau_n)}-1\right| < \epsilon$. Therefor, for $n > N_{\epsilon}$, $\mathbb{P}_n(L(h, b, \tau_n) - L(h, b, \tau)) < M\epsilon$, which shows that (IV) converges to zero as $n$ goes to infinity. The desired results are concluded.
\end{proof}

\subsection{Numerical Experiment in Observational Study}
\label{A.10}

In this section, we extend the simulation studies in Section \ref{sec: sim} to observational studies, where a propensity score model $\mathbb{P}(\mathbf{A}|\mathbf{X})$ is controlling the treatment assignment. Specifically, the propensity score model we adopt is defined as
\begin{align}
    \mathbb{P}(\tilde{A}_j|\mathbf{X}) = \frac{\exp(j\cdot\mathbf{X}^T\tau)}{\sum_{j}\exp(j\cdot\mathbf{X}^T\tau)},
\end{align}
where $\tau = (-0.5, -0.4, ..., -0.1, 0.1, ..., 0.4, 0.5) \in \mathbb{R}^{10}$, and $\tilde{A}_j \in \{1, 2, ..., 2^{K}\}$ is the categorical coding of combination treatment $\mathbf{A}\in\{-1, 1\}^{K}$. In our algorithm, we first estimate the propensity score using the multinomial logistic regression \citep{hastie2009elements} and then plug it into (\ref{empirical_loss}) to estimate the ITR. 

In this simulation, all other data generating processes including covariates distribution, treatment effects, and sample sizes are identical to the settings in (\ref{sec: sim}). Table \ref{tab: sim_os_results_value} and \ref{tab: sim_os_results_accuracy} present the evaluation and comparison of our methods with competing methods, which demonstrate our method can still outperform competing methods in the observational study settings.

\begin{table}[!ht]
    \centering
    \scriptsize
    \begin{tabular}{c|c|ccc}
    \hline
    Setting & Method & 400 & 800 & 2000\\
    \hline
    \multirow{7}{*}{1}&\textbf{MLRWL-Linear} & \textbf{4.112(0.144)} & \textbf{4.398(0.082)} & \textbf{4.437(0.071)} \\
    &\textbf{MLRWL-Kernel} & 3.934(0.137) & 4.218(0.073) & 4.357(0.059)\\
    &OWL-DL & 4.010(0.118) & 4.100(0.109) & 4.201(0.089) \\
    &$L_1$-PLS & 4.057(0.109) & 4.148(0.085) & 4.265(0.094) \\
    &OWL-MD & 3.660(0.176) & 3.806(0.125) & 3.927(0.118) \\
    &MOWL-Linear & 3.132(0.184) & 3.180(0.104) & 3.280(0.097) \\
    &MOWL-Kernel & 2.892(0.184) & 3.002(0.224) & 3.273(0.110) \\
    \hline
    \multirow{7}{*}{2}&\textbf{MLRWL-Linear} & 1.382(0.055) & 1.420(0.049) & 1.427(0.047) \\
    &\textbf{MLRWL-Kernel} & \textbf{1.836(0.079)} & \textbf{1.948(0.062)} & \textbf{2.080(0.056)} \\
    &OWL-DL & 1.678(0.098) & 1.701(0.093) & 1.702(0.091) \\
    &$L_1$-PLS & 1.689(0.089) & 1.707(0.099) & 1.724(0.083) \\
    &OWL-MD & 1.657(0.130) & 1.684(0.120) & 1.699(0.079) \\
    &MOWL-Linear & 1.771(0.159) & 1.893(0.118) &  1.938(0.108)\\
    &MOWL-Kernel & 1.798(0.039) & 1.904(0.039) & 1.969(0.032) \\
    \hline
    \multirow{7}{*}{3}&\textbf{MLRWL-Linear} & 4.413(0.253) & 4.618(0.192) & 4.660(0.113) \\
    &\textbf{MLRWL-Kernel} & \textbf{4.730(0.086)} & \textbf{4.734(0.084)} & \textbf{4.736(0.075)} \\
    &OWL-DL & 4.302(0.200) & 4.602(0.198) & 4.639(0.187) \\
    &$L_1$-PLS & 4.205(0.216) & 4.188(0.186) & 4.219(0.143) \\
    &OWL-MD & 4.421(0.183) & 4.499(0.290) & 4.501(0.214) \\
    &MOWL-Linear & 4.609(0.169) & 4.600(0.101) & 4.602(0.086) \\
    &MOWL-Kernel & 4.712(0.090) & 4.703(0.077) & 4.721(0.080) \\
    \hline
    \end{tabular}
    \caption{Simulation studies: mean and standard error of the value function under the proposed method with linear and nonlinear decision rules, and five competing methods: the outcome weighted learning with deep learning (OWL-DL, \citep{liang2018estimating}), the $L_1$ penalized least square ($L_1$-PLS, \citep{qian2011performance}), the outcome weighted learning with multinomial deviance (OWL-MD, \citep{huang2019multicategory}), and the multicategory outcome weighted learning with linear and kernel functions (MOWL-Linear and MOWL-Kernel, \citep{zhang2020multicategory}). Higher value is better.}
    \label{tab: sim_os_results_value}
\end{table}

\begin{table}[!ht]
    \centering
    \scriptsize
    \begin{tabular}{c|c|ccc}
    \hline
    Setting & Method & 400 & 800 & 2000 \\
    \hline
    \multirow{7}{*}{1}&\textbf{MLRWL-Linear} & \textbf{0.773(0.057)} & \textbf{0.861(0.028)} & \textbf{0.893(0.019)} \\
    &\textbf{MLRWL-Kernel} & 0.691(0.047) & 0.764(0.025) & 0.803(0.018)  \\
    &OWL-DL & 0.633(0.047) & 0.649(0.039) & 0.672(0.030) \\
    &$L_1$-PLS & 0.653(0.025) & 0.674(0.014) & 0.694(0.014) \\
    &OWL-MD & 0.615(0.055) & 0.645(0.037) & 0.662(0.028) \\
    &MOWL-Linear & 0.466(0.032) & 0.477(0.023) & 0.497(0.018) \\
    &MOWL-Kernel & 0.364(0.056) & 0.382(0.050) & 0.464(0.028) \\
    \hline
    \multirow{7}{*}{2}&\textbf{MLRWL-Linear} & 255(0.021) & 0.266(0.013) & 0.267(0.011) \\
    &\textbf{MLRWL-Kernel} & \textbf{0.473(0.017)} & \textbf{0.522(0.019)} & \textbf{0.596(0.013)} \\
    &OWL-DL & 0.334(0.038) & 0.342(0.030) & 0.350(0.032) \\
    &$L_1$-PLS & 0.326(0.027) & 0.328(0.019) & 0.334(0.012) \\
    &OWL-MD & 0.326(0.022) & 0.326(0.016) & 0.313(0.022) \\
    &MOWL-Linear & 0.354(0.026) & 0.367(0.018) & 0.368(0.011) \\
    &MOWL-Kernel & 0.386(0.006) & 0.387(0.007) & 0.400(0.024) \\
    \hline
    \multirow{7}{*}{3}&\textbf{MLRWL-Linear} & 0.379(0.180) & 0.565(0.167) & 0.600(0.153) \\
    &\textbf{MLRWL-Kernel} & \textbf{0.721(0.072)} & \textbf{0.723(0.052)} & \textbf{0.742(0.013)} \\
    &OWL-DL & 0.493(0.054) & 0.542(0.049) & 0.608(0.049) \\
    &$L_1$-PLS & 0.180(0.028) & 0.170(0.023) & 0.172(0.014) \\
    &OWL-MD & 0.386(0.058) & 0.422(0.062) & 0.467(0.067) \\
    &MOWL-Linear & 0.518(0.092) & 0.510(0.049) & 0.504(0.043) \\
    &MOWL-Kernel & 0.718(0.020) & 0.726(0.036) & 0.733(0.048) \\
    \hline
    \end{tabular}
    \caption{Simulation studies: mean and standard error of the accuracy under the proposed method with linear and nonlinear decision rules, and five competing methods: the outcome weighted learning with deep learning (OWL-DL, \citep{liang2018estimating}), the $L_1$ penalized least square ($L_1$-PLS, \citep{qian2011performance}), the outcome weighted learning with multinomial deviance (OWL-MD, \citep{huang2019multicategory}), and the multicategory outcome weighted learning with linear and kernel functions (MOWL-Linear and MOWL-Kernel, \citep{zhang2020multicategory}).}
    \label{tab: sim_os_results_accuracy}
\end{table}

\end{appendix}

\newpage
\begin{acks}[Acknowledgments]
The authors would like to thank the anonymous referees, the Associate Editor and the Editor for their constructive comments that improved the quality of this paper.
\end{acks}

\begin{funding}
This work is supported by National Science Foundation Grants DMS 2210640 and DMS 1952406.
\end{funding}

\bibliographystyle{imsart-number} 
\bibliography{reference.bib}       

\end{document}